\documentclass[pdflatex,sn-mathphys-num]{sn-jnl}
\usepackage{amsmath,amssymb,amsfonts}
\usepackage{amsthm}
\usepackage{amscd}
\usepackage{cases}
\usepackage{cancel}
\usepackage{breqn}
\usepackage{tikz}
\usepackage{cleveref}
\usepackage{textcomp}
\usepackage{xcolor}

\usetikzlibrary{arrows.meta}
\def\sc{\scriptstyle}

\def\de{\delta}

\def\la{\lambda}
\def\a{\alpha}

\def\p{\partial}
\def\vs{\vspace*}

\def\A{\mathcal{A}}

\def\O{\mathcal{O}}
\def\B{\mathcal{B}}
\def\C{\mathbb{C}}

\def\Cend{{\rm {Cend}}}
\def\ad{{\rm {ad}}}
\def\id{{\rm {id}}}

\def\ad{{\rm {ad}}}
\def\gc{{\rm {gc}}}

\def\id{{\rm id}}
\def\Vir{\hbox{Vir}}

\def\Cend{{\rm {Cend}}}

\def\vs{\vspace*}

\newcommand{\circlednum}[1]{\tikz[baseline=(char.base)]{\node[shape=circle,draw,inner sep=0.5pt](char){#1};}}
\numberwithin{equation}{section}
\newtheorem{theo}{Theorem}[section]
\newtheorem{defi}[theo]{Definition}
\newtheorem{coro}[theo]{Corollary}
\newtheorem{lemm}[theo]{Lemma}
\newtheorem{prop}[theo]{Proposition}

\newtheorem{ex}[theo]{Example}

\newtheorem{remark}[theo]{Remark}

\allowdisplaybreaks[4]

\begin{document}

\title[Compatible Lie conformal algebras and bialgebras]{Compatible Lie conformal bialgebras}

\author*[1]{\fnm{Lamei} \sur{Yuan}}\email{lmyuan@hit.edu.cn}

\author[1]{\fnm{Jiansen} \sur{Zhao}}\email{24S012013@stu.hit.edu.cn}

\affil*[1]{\orgdiv{School of Mathematics}, \orgname{Harbin Institute of Technology}, \orgaddress{\city{Harbin}, \postcode{150001}, \country{China}}}

\abstract{We introduce and study compatible Lie conformal bialgebras as conformal counterparts of compatible Lie bialgebras. Such a structure consists of two compatible Lie conformal brackets and two compatible conformal cobrackets on the same $\C[\partial]$-module, and every simultaneous linear combination of them is again a Lie conformal bialgebra. We develop representations and matched pairs for compatible Lie conformal algebras, introduce compatible Lie conformal coalgebras, and establish their duality through the conformal dual in the finite case. For $\C[\p]$-modules that are free of finite rank, we prove the equivalence among compatible Lie conformal bialgebras, standard compatible conformal Manin triples and matched pairs. 
In the coboundary case, we characterize the tensors $r$ that determine compatible Lie conformal bialgebras. The characterization requires the symmetric part of $r$ to be invariant with respect to both brackets and three conformal Yang--Baxter conditions to hold. The first two conditions correspond to the two brackets separately, whereas the third is the compatible conformal Yang--Baxter condition.
For comparison, we introduce the compatible conformal classical Yang--Baxter equation (CYBE) and show that each of its solutions satisfies these three conditions, while the converse fails.
One explicit example shows that the first two conditions do not imply the third. Another shows that all three conformal Yang-Baxter conditions may hold even though $r$ is not a solution of compatible conformal CYBE.}



\maketitle

\section{Introduction}

Lie conformal algebras were introduced by Kac as an algebraic language for the singular part of the operator product expansion of chiral fields in two-dimensional conformal field theory \cite{KAC}. They may be viewed as formal-distribution analogues of Lie algebras: the ordinary Lie bracket is replaced by a $\lambda$-bracket, and the translation operator is encoded in the $\C[\partial]$-module structure. This framework has become a standard setting for the structure theory, representation theory and cohomology of conformal algebras \cite{BKV,CK,DK,DK3}.

	For ordinary Lie algebras, one of the basic sources of bialgebra structures is the classical Yang--Baxter equation. If $\mathfrak g$ is a Lie algebra and $r\in\mathfrak g\otimes\mathfrak g$, the tensor form is
$
	[r_{12},r_{13}]+[r_{12},r_{23}]+[r_{13},r_{23}]=0.
$	
	Equivalently, if $r=\sum_i x_i\otimes y_i$, this equation reads
	\[
		\sum_{i,j}\Big(
		[x_i,x_j]\otimes y_i\otimes y_j
		-x_i\otimes [x_j,y_i]\otimes y_j
		-x_i\otimes x_j\otimes [y_j,y_i]
		\Big)=0.
	\]
	Solutions give rise to triangular Lie bialgebras and are closely related to Poisson--Lie groups, quantum groups and integrable systems \cite{Drinfeld,STS,BD}. The coboundary formula $\delta(x)=x\cdot r$ is the mechanism which connects this equation with Lie bialgebras and Manin triples.

Liberati constructed the conformal counterpart of this picture by introducing Lie conformal coalgebras and Lie conformal bialgebras, together with conformal Manin triples, Drinfeld doubles and a conformal classical Yang--Baxter equation \cite{Li}. In this setting, tensor equations cannot be copied verbatim from Lie algebra theory. The bracket carries a formal parameter, and the position of the translation operator in the tensor factors becomes part of the formula. For example, if $r=\sum_i a_i\otimes b_i$ in a Lie conformal algebra $\A$, the conformal Yang--Baxter expression is
\[
\begin{aligned}
\relax[\![r,r]\!]
=\sum_{i,j}\Big([a_i{}_\mu a_j]\otimes b_i\otimes b_j\big|_{\mu=1\otimes\partial\otimes1}
-a_i\otimes[a_j{}_\mu b_i]\otimes b_j\big|_{\mu=1\otimes1\otimes\partial}
-a_i\otimes a_j\otimes[b_j{}_\mu b_i]\big|_{\mu=1\otimes\partial\otimes1}\Big),
\end{aligned}
\]
and the conformal classical Yang--Baxter equation is imposed modulo the total translation operator. Further work relates this equation to conformal $\mathcal O$-operators, conformal $S$-equations and bialgebra structures of associative, Gel'fand--Dorfman and Poisson conformal algebras \cite{HB1,HB2,HBG,HB3,GR1,GR2,YL,YLZ}.

The question considered here comes from compatibility. In many algebraic problems one studies not a single isolated operation, but a family of linear combinations of two operations. For instance, two Lie brackets are compatible if every linear combination of them is again a Lie bracket. This idea underlies compatible Poisson structures, bi-Hamiltonian systems and several deformation-theoretic constructions \cite{B1,B2,GS1,GS2,GS3,LSB}. Wu and Bai developed the corresponding bialgebra theory for compatible Lie algebras and obtained a compatible form of the classical Yang--Baxter equation in the coboundary case \cite{WB}.

It is natural to ask whether the same picture survives for Lie conformal algebras. Suppose that a $\C[\partial]$-module carries two Lie conformal brackets
$[\cdot_\lambda\cdot]^1$, $[\cdot_\lambda\cdot]^2$
and two conformal cobrackets
$\delta_1$, $\delta_2$.
When does every linear combination $k_1[\cdot_\lambda\cdot]^1+k_2[\cdot_\lambda\cdot]^2$, $k_1\delta_1+k_2\delta_2$
define a Lie conformal bialgebra? The answer is not obtained by a formal translation of the Lie algebra case. Besides the compatible Jacobi and compatible co-Jacobi identities, one has to keep track of conformal sesquilinearity, the conformal dual module and the precise placement of $\partial$ in the tensor substitutions. These are exactly the points where the conformal theory differs from the ordinary Lie algebra theory.

This paper studies this compatibility problem for Lie conformal bialgebras. We first introduce compatible Lie conformal algebras, together with their representations and matched pairs, and then define compatible Lie conformal coalgebras and compatible Lie conformal bialgebras. For $\C[\partial]$-modules that are free of finite rank, we prove that compatible Lie conformal bialgebras, standard compatible conformal Manin triples and matched pairs of compatible Lie conformal algebras are equivalent. This is the compatible analogue of Liberati's correspondence between Lie conformal bialgebras and conformal Manin triples. Along the way we establish a duality between finite compatible Lie conformal algebras and compatible Lie conformal coalgebras through the conformal dual, and we show that the conformal dual of a compatible Lie conformal bialgebra of finite rank is again a compatible Lie conformal bialgebra.

The coboundary case leads to a further question. Let $(\A,[\cdot_\lambda\cdot]^1,[\cdot_\lambda\cdot]^2)$ be a compatible Lie conformal algebra and let $r\in\A\otimes\A$. The same tensor $r$ gives two coboundary cobrackets, one through each bracket, and we determine exactly when the resulting data form a compatible Lie conformal bialgebra: the symmetric part of $r$ must be invariant for both brackets, $r$ must satisfy the conformal Yang--Baxter condition for each bracket, and it must satisfy a compatible conformal Yang--Baxter condition involving both brackets (see Theorem \ref{com YB}). We also introduce the compatible conformal CYBE, and show that every solution satisfies these three conformal Yang--Baxter conditions. Two examples delimit the theory: in one, the conformal Yang--Baxter conditions for the two brackets hold while the compatible conformal Yang--Baxter condition fails, so the compatible condition is an independent requirement; in another, built from a Schr\"odinger--Virasoro type compatible Lie conformal algebra, the three conformal Yang--Baxter conditions all hold although $r$ is not a solution of the compatible conformal CYBE.

The paper is organized as follows. Section~2 recalls Lie conformal algebras, their representations and conformal duals. Section~3 introduces compatible Lie conformal algebras, their representations and matched pairs. Section~4 treats compatible Lie conformal coalgebras and their duality with compatible Lie conformal algebras in the finite free case. Section~5 proves the equivalence among compatible Lie conformal bialgebras, standard compatible conformal Manin triples and matched pairs of compatible Lie conformal algebras. Section~6 is devoted to the coboundary construction and the compatible conformal classical Yang--Baxter equation.

Throughout the paper, all vector spaces, linear maps and tensor products are over $\C$. For a vector space $A$, the space of polynomials in $\lambda$ with coefficients in $A$ is denoted by $A[\lambda]$.


\section{Preliminaries}
\quad In this section, we recall some basic definitions and notations. The material can be found in \cite{DK,DK3,HB1,HB2,KAC,Li}.
\begin{defi}\label{defi1}
{\rm A { conformal algebra} $\mathcal{A}$ is a $\mathbb{C}[\partial]$-module  endowed with a $\mathbb{C}$-bilinear map
\begin{equation*}
\mathcal{A}\otimes \mathcal{A} \rightarrow \mathcal{A} [\lambda], ~~~~~~a\otimes b\mapsto a_\lambda b,
\end{equation*}
satisfying conformal sesquilinearity:
\begin{eqnarray}\label{sesquilinearity1}
(\partial a)_\lambda b=-\lambda a_\lambda b,~~~
a_\lambda \partial b=(\partial+\lambda)a_\lambda b,  ~~ \forall ~~ a, b\in \mathcal{A}.
\end{eqnarray} 
If, in addition, it satisfies associativity:
\begin{eqnarray}\label{ASS}
(a_\lambda b)_{\lambda+\mu}c=a_\lambda (b_\mu c), ~~ \forall ~~ a, b, c\in \mathcal{A},
\end{eqnarray}
then $\mathcal{A}$ is called an {associative conformal algebra}.
}
\end{defi}

A conformal algebra is called {\it finite} if it has finite rank as $\C[\p]$-module.

\begin{defi}\rm 
A { Lie conformal algebra} (LCA) $\A$ is a $\C[\partial]$-module endowed with a $\C$-bilinear map
\[
 \A\otimes \A\rightarrow \A[\lambda],\ \  a\otimes b \mapsto [a_\lambda b],
\]
called the $\la$-bracket, and
satisfying the following axioms (for all $a, b, c\in \A$):
\begin{align}
\mbox{conformal\  sesquilinearity:}~~ [\partial a_\lambda b]&=-\lambda[a_\lambda b],\ \ [ a_\lambda \partial b]=(\partial+\lambda)[a_\lambda b], \label{Lc1}\\
 \ \mbox{skew-symmetry:}~~ {[a_\lambda b]} &= -[b_{-\lambda-\partial}a], \ \label{Lc2}\\
 \mbox{Jacobi \ identity:}~~ {[a_\lambda[b_\mu c]]}&=[[a_\lambda b]_{\lambda+\mu
}c]+[b_\mu[a_\lambda c]]. \label{Lc3}
\end{align}
\end{defi}
\begin{defi} \rm Let $U$ and $V$ be two $\mathbb{C}[\partial]$-modules.
A conformal homomorphism (or {conformal linear map}) from $U$ to $V$ is a $\mathbb{C}$-linear map
$f_\lambda: U\rightarrow V[\lambda]$, such that
\begin{eqnarray}
f_\lambda(\partial u)=(\partial+\lambda)f_\lambda (u), ~~ \forall ~~ u\in U.
\end{eqnarray}
The vector space of all such maps, denoted by ${\rm Chom}(U,V)$, is a $\C[\partial]$-module via
\begin{align}
(\partial f)_\la=-\la f_\la, \ \forall\ \ f_\la\in {\rm Chom}(U,V).
\end{align}
\end{defi}
Define the {\it conformal dual} of a $\C[\partial]$-module $M$ as $M^{*c}={\rm Chom}(M,\C)$, where $\C$ is viewed as the trivial $\C[\partial]$-module, that is,
\begin{align}
 M^{*c}= \{f:M\rightarrow \C[\lambda] \mid  \mbox{$f$ is $\C$-linear and } f_{\lambda}(\partial v)=\lambda f_{\lambda}v, ~~\forall~v\in M  \}.
\end{align}
	If $M$ is a finitely generated $\C[\partial]$-module, then ${\rm Cend}(M):={\rm Chom}(M, M)$
	is an associative conformal algebra via:
	\begin{align*}
	(f_\la g)_\mu v=f_\la(g_{\mu-\la}v), ~~ \mbox{for~all} ~~v\in M, ~ f,g\in {\rm Cend}(M).
	\end{align*}
	Hence ${\rm Cend}(M)$ becomes a Lie conformal algebra, denoted by ${\rm gc}(M)$ and called the {\it general Lie conformal algebra} on $M$, with respect to
	the following $\la$-bracket:
	\begin{align*}
	[f_\la g]_\mu v=f_\la(g_{\mu-\la}v) -g_{\mu-\la}(f_\la v),~~ \mbox{for} ~~v\in M, ~ f,g\in {\rm Cend}(M).
	\end{align*}

\begin{defi}\rm
	Let $U$ and $V$ be two $\mathbb{C}[\partial]$-modules. For a $\mathbb{C}[\p]$-module homomorphism $\phi \colon U \to V$, we denote the dual $\mathbb{C}[\p]$-module homomorphism $\phi^* \colon V^{*c} \to U^{*c}$ given by
	\[
	\phi^*(f)_\lambda (u) = f_\lambda(\phi(u)), \quad f \in V^{*c}, u \in U.
	\]
\end{defi}

If $U$ and $V$ are conformal modules of a Lie conformal algebra $R$, then ${\rm Chom}(U,V )$ also has a conformal $R$-module structure defined by
\begin{align}
(a_\la \varphi)_\mu u= a_\la^V (\varphi_{\mu-\la}u)-\varphi_{\mu-\la}(a_\la^U u),
\end{align}
where $a\in R$, $\varphi\in {\rm Chom}(U,V)$ and $u\in U$. Therefore, one particular case is the contragradient
conformal $R$-module $U^{*c} = {\rm Chom}(U,\C)$, also called the {\it dual module}, where $\C$ is viewed as the trivial $R$-module and $\C[\partial]$-module.

\begin{defi} \rm 
	A representation $(M,\rho)$ of a  Lie conformal algebra $ (\A,{\left [\cdot  _{\lambda}\cdot\right ]}) $ is pair of a finite $\mathbb{C}[\partial]$-module $M$ and a $\mathbb{C}$-linear map $\rho: \A \rightarrow \gc(M)$, such that
	\begin{align}
		\rho\left ( a \right )_\lambda\rho\left (b   \right )_\mu-\rho\left ( b \right )_\mu\rho\left (a  \right )_\lambda&=\left [ \rho\left (  a\right )_\lambda\rho\left ( b \right ) \right ]_{\lambda+\mu}=\rho\left (\left [ a_\lambda b \right ]  \right )_{\lambda+\mu},\\		
		\rho\left (  \partial\left ( a \right )\right )_\lambda&=-\lambda\rho\left ( a \right )_\lambda , \forall a,b\in \A.
	\end{align} That is, $\rho$ is a homomorphism of Lie conformal algebras from $\A$ to $\gc\left(M  \right )$. 
\end{defi}

Let $(M,\rho)$ be a representation of a  Lie conformal algebra $(\A,{\left [\cdot  _{\lambda}\cdot\right ]})$. Then
$\A \oplus M$ has a  $\C[\p]$-module structure given by 
\[
\p(a+v)=\p a+\p v, ~~ \forall ~~ a\in\A,~v\in M,
\]
and   $\A \oplus M$ carries a  Lie conformal algebra structure via
\[
[\!\![(x+u)_\la (y+v)]\!\!] = [x_\la y] + \rho(x)_\la v - \rho(y)_{-\la-\p}u, ~~ \forall ~~ x,y\in\A,~u,v\in M.
\]
 We call $(\A \oplus M, [\!\![\cdot_\la \cdot]\!\!])$ the {\it semi-direct product} of $(\A,[\cdot_\la \cdot])$ and $(M,\rho)$, and we  denote it by $\A \ltimes M$.

Moreover, define $\rho^*:\A\rightarrow \Cend(M^{*c})$  by 
\begin{align}\label{dual representation}
	(\rho^*(a)_\lambda \phi)_\mu v = -\phi_{\mu-\lambda}(\rho(a)_\lambda v),\quad \forall a\in \A,\phi \in M^{*c},v\in M.  
\end{align}
The defining identity for $(M^{*c},\rho^*)$ follows from the representation identity for
$(M,\rho)$ by applying \eqref{dual representation}. This representation is called the
{\it dual representation} of $(M,\rho)$.

For example, the map $\ad: \A\rightarrow \Cend(\A)$ defined by
\[
\ad(a)_\la b = [a_\la b],~~ \forall~ a,b \in \A,
\]
makes $\A$ a representation of itself, called the {\it adjoint representation}. Hence, $(\A^{*c},\ad^{*})$ is also a representation of $\A$, called the {\it coadjoint representation.}

\begin{ex}\label{ex:Vir-rank-one-module}
Let $\Vir=\C[\partial]L$ be the Virasoro Lie conformal algebra with
\[
[L_\lambda L]=(\partial+2\lambda)L.
\]
For $\Delta,\alpha\in\C$, the free rank-one $\C[\partial]$-module
$M_{\Delta,\alpha}=\C[\partial]v$ is a $\Vir$-module defined by
\[
L_\lambda v=(\partial+\Delta\lambda+\alpha)v.
\]
When $\alpha=0$, we call $\Delta$ the conformal weight of $v$.
\end{ex}

\begin{defi}\label{RBLCA}{\rm 
		Let $(\A,[\cdot_\lambda\cdot])$ be a Lie conformal algebra and $\a \in \C$. If $T:\A \rightarrow \A$ is a $ \mathbb{C} [\partial] $-module homomorphism satisfying
		\begin{align}
			[T(a)_\lambda T(b)] = T([a_\lambda T(b)]+[T(a)_\lambda b]+\alpha [a_\lambda b]),~~\forall~ a,b\in \A,
		\end{align}
		then $T$ is called a Rota-Baxter operator of weight $ \alpha $ on $\A$. In this case, $(\A,[\cdot_\lambda\cdot],T)$ is called a Rota-Baxter Lie conformal algebra of weight $ \alpha $.}
\end{defi}

\begin{prop} \label{descendent}
		Let $(\A,[\cdot_\lambda \cdot])$ be a Lie conformal algebra and $T:\A \rightarrow \A$ a Rota-Baxter operator of weight $\alpha$
	on $\A$. Then there is a new $\la$-bracket $[\cdot_\la \cdot]^T$ on $\A$ defined by 
	\begin{align}
		[x_\la y]^T = [T(x)_\lambda y]+[x_\lambda T(y)]+\alpha [x_\lambda y], \quad \forall x,y\in \A.
	\end{align}
	The Lie conformal algebra $ (\A,[\cdot_\la \cdot]^T) $ is called the {\it descendant Lie conformal algebra} and simply denoted by $\A_T$. By the Rota-Baxter identity, $T$ is a Lie conformal algebra homomorphism from $(\A,[\cdot_\la\cdot]^T)$ to $(\A,[\cdot_\la\cdot])$:
	\begin{align}
		T\big([x_\la y]^T\big) = [T(x)_\lambda T(y)], ~~\forall~ x,y\in \A.
	\end{align}
\end{prop}


\section{Compatible Lie conformal algebras}

In this section, we introduce compatible Lie conformal algebras and their representations, extending the classical compatibility condition to the conformal setting. We then define and study matched pairs of compatible Lie conformal algebras, which serve as a key structure connecting compatible Lie conformal bialgebras and comformal Manin triples.

\subsection{Definitions of compatible Lie conformal algebras  and representations}

By analogy with the case of Lie algebras, the following conclusions hold for Lie conformal algebras:
\begin{prop}\label{compatible prop}
	Let $(\A,[\cdot _\la \cdot]^1)$ and $(\A,[\cdot _\la \cdot]^2)$ be two Lie conformal algebras over the same $\C[\p]$-module $\A$. Then the following conditions are equivalent:
	\begin{itemize}[(c)]
		\item [(i)] $(\A,[\cdot_\la\cdot]^1+[\cdot_\la\cdot]^2)$ is a Lie conformal algebra;
		\item [(ii)] For all $k_1,k_2 \in \C$, $(\A,k_1[\cdot_\la\cdot]^1+k_2[\cdot_\la\cdot]^2)$ is a Lie conformal algebra;
		\item [(iii)] The compatible Jacobi identity holds for all $x,y,z\in \A$:
		\begin{equation}\label{compatible Jacobi}
			\begin{aligned}
				\relax [x_{\lambda} [y_\mu z]^1]^2
				+ [x_\lambda {\left [y_\mu z \right ]}^2]^1
				- [{\left [x_\lambda y \right ]}^{1}_{\lambda+\mu}z]^2
				- [y_\mu{ [x_\lambda z]}^1]^2
				- [{\left [x_\lambda y \right ]}^{2}_{\lambda+\mu}z]^1
				- [y_\mu{\left [x_\lambda z \right ]}^2]^1=0.
			\end{aligned}
		\end{equation}
	\end{itemize}
\end{prop}

With this we have the definition of compatible Lie conformal algebras.
\begin{defi}\label{compatible def}{\rm
		A compatible Lie conformal algebra is a triple $(\A,[\cdot_\la \cdot]^1,[\cdot _\la \cdot]^2)$, where $(\A,[\cdot_\la \cdot]^1)$ and $(\A,[\cdot _\la \cdot]^2)$ are Lie conformal algebras satisfying any of the three equivalent conditions in Proposition \ref{compatible prop}. }
\end{defi}
If $(\A,[\cdot _\la \cdot]^1,[\cdot_\la \cdot]^2)$ is a compatible LCA, then we say  $(\A,[\cdot_\la \cdot]^1)$ and $(\A,[\cdot _\la \cdot]^2)$ are compatible, and we also say  $(\A,[\cdot_\la \cdot]^1,[\cdot _\la \cdot]^2)$ is the compatible LCA of  $(\A,[\cdot_\la \cdot]^1)$ and $(\A,[\cdot _\la \cdot]^2)$. Sometimes, we will simply denote the compatible LCA of $(\A,[\cdot _\la \cdot]^1, [\cdot _\la \cdot]^2)$ by $\tilde{\A}$.

\begin{ex}(\cite{WY})\label{exW1}
	For $b\in \C$, the $W(b)$ Lie conformal algebra is a free $\C[\p]$-module generated by $L$ and $H$, satisfying the following $\la$-brackets:
	\begin{align*}
		[L_\la L] = (\p + 2\la)L,\quad [L_\la H] = (\p + (1-b)\la)H,\quad [H_\la L] = (-b\p + (1-b)\la)H,\quad [H_\la H] = 0.
	\end{align*}
		It is straightforward to verify that for any $b_1,b_2\in \C$, the Lie conformal algebras $W(b_1)$ and $W(b_2)$ are compatible.
		In particular, $W(0)$ is the Heisenberg--Virasoro conformal algebra, and it is compatible with $W(b)$ for every $b\in\C$.
\end{ex}

\begin{ex}\label{ex-Rp-compatible-LCA}
	Let $R=\C[\p]a\oplus \C[\p]b$.
	Define two $\lambda$-brackets on $R$ by
	\[
		[a_\lambda b]^1=b,\qquad [a_\lambda a]^1=[b_\lambda b]^1=0,
	\]
	and
	\[
		[a_\lambda b]^2=\lambda^2b,\qquad [a_\lambda a]^2=[b_\lambda b]^2=0,
	\]
	extended by skew-symmetry and conformal sesquilinearity. These are the rank two solvable
	Lie conformal algebras $R_p$ \cite[Example 2.18]{Li}, with $p(\lambda)=1$ and
	$p(\lambda)=\lambda^2$, respectively. It remains only to verify the compatible Jacobi identity
	\eqref{compatible Jacobi}. Up to skew-symmetry and conformal sesquilinearity, the only
	nonzero case is $(x,y,z)=(a,a,b)$, for which the left-hand side of \eqref{compatible Jacobi}
	is
	\[
		\lambda^2b+\mu^2b-\mu^2b-\lambda^2b=0.
	\]
	Therefore $(R,[\cdot_\lambda\cdot]^1,[\cdot_\lambda\cdot]^2)$ is a compatible Lie conformal algebra.
\end{ex}

The following example will be used as the algebraic datum for the coboundary construction in
Example \ref{ex-sv-nonzero-mixed-cybe}.

\begin{ex}\label{ex-sv-compatible-LCA}
	Let
	\[
		\A=\C[\p]L\oplus \C[\p]P\oplus \C[\p]Q\oplus \C[\p]M_1\oplus \C[\p]M_2 .
	\]
	The construction is motivated by the Schr\"odinger--Virasoro Lie conformal algebra introduced
	in \cite{SY}. For $s=1,2$, define two $\lambda$-brackets $[\cdot_\la\cdot]^1$ and
	$[\cdot_\la\cdot]^2$ on $\A$ by
	\[
		[L_\la L]^s=(\p+2\la)L,\qquad
		[L_\la P]^s=(\p+\frac{1}{2}\la)P,\qquad
		[L_\la Q]^s=(\p+\frac{1}{2}\la)Q,\qquad
		[L_\la M_i]^s=\p M_i,\quad i=1,2,
	\]
	and
	\[
		[P_\la Q]^1=M_1,\qquad [P_\la Q]^2=M_2.
	\]
	All other brackets between generators are zero, and the brackets are extended by conformal
	sesquilinearity and skew-symmetry. Thus $P$ and $Q$ have conformal weight $\frac12$, while
	$M_1$ and $M_2$ have conformal weight $0$.
	
	Set
	\[
		\mathfrak n=\C[\p]P\oplus\C[\p]Q\oplus\C[\p]M_1\oplus\C[\p]M_2 .
	\]
	Then $\mathfrak n$ is a conformal ideal with respect to each of the two brackets. More precisely,
	\[
		[\mathfrak n_\lambda\mathfrak n]^s=\C[\p]M_s,\qquad
		[\mathfrak n_\lambda\C[\p]M_s]^s=0,\qquad s=1,2.
	\]
	Thus $\mathfrak n$ is two-step nilpotent with respect to each bracket; equivalently,
	\[
		[\mathfrak n_\lambda[\mathfrak n_\mu\mathfrak n]^s]^s=0,\qquad s=1,2.
	\]
	Moreover,
	$\C[\p]M_1\oplus\C[\p]M_2$ is contained in the center of $\mathfrak n$ with respect to both
	brackets, but it is not central in the whole Lie conformal algebra, since the Virasoro
	generator acts nontrivially on $M_1$ and $M_2$.
	
	For each fixed $s$, the brackets involving $L$ define the usual Virasoro action on
	$P,Q,M_1$ and $M_2$, with conformal weights $\frac12,\frac12,0,0$, respectively. Therefore the Jacobi identities
	involving only $L$ and one of these generators hold. Since the brackets among $P,Q,M_1$ and $M_2$ are
	two-step nilpotent, the only remaining Jacobi identity is the one for $L,P,Q$.
	Since $[P_\mu Q]^s=M_s$, we have
	\[
		[L_\lambda [P_\mu Q]^s]^s=\p M_s.
	\]
	On the other hand,
	\begin{align*}
		[[L_\lambda P]^s{}_{\lambda+\mu}Q]^s+[P_\mu [L_\lambda Q]^s]^s
		=(-\frac12\lambda-\mu)M_s+(\p+\mu+\frac12\lambda)M_s
		=\p M_s.
	\end{align*}
	Thus each $[\cdot_\lambda\cdot]^s$ is a Lie conformal bracket.
	
	It remains to verify the compatible Jacobi identity \eqref{compatible Jacobi}. Since the two
	brackets have the same Virasoro action on $P,Q,M_1$ and $M_2$, the triples
	involving only $L$ and one of $P,Q,M_1,M_2$ reduce to the Jacobi identities for the individual
	brackets checked above. If no $L$ occurs, then the two-step nilpotence of the internal part makes every
	term in \eqref{compatible Jacobi} vanish. Thus the only case which is not automatic is the
	triple $(L,P,Q)$, up to skew-symmetry and conformal sesquilinearity.
	
	For $x=L$, $y=P$ and $z=Q$, the first two terms in \eqref{compatible Jacobi} are
	\[
		[L_\lambda [P_\mu Q]^1]^2+[L_\lambda [P_\mu Q]^2]^1
		=\p M_1+\p M_2.
	\]
	The remaining four terms are
	\begin{align*}
		&-[[L_\lambda P]^1{}_{\lambda+\mu}Q]^2-[P_\mu [L_\lambda Q]^1]^2
		-[[L_\lambda P]^2{}_{\lambda+\mu}Q]^1-[P_\mu [L_\lambda Q]^2]^1\\
		=&\left(\frac12\lambda+\mu\right)M_2-\left(\p+\mu+\frac12\lambda\right)M_2
		+\left(\frac12\lambda+\mu\right)M_1-\left(\p+\mu+\frac12\lambda\right)M_1\\
		=&-\p M_1-\p M_2.
	\end{align*}
	Hence the sum of the six terms in \eqref{compatible Jacobi} is zero. Therefore
	$(\A,[\cdot_\lambda\cdot]^1,[\cdot_\lambda\cdot]^2)$ is a compatible Lie conformal algebra.
\end{ex}

\begin{ex}
	Let $(\A,[\cdot_\la\cdot])$ be a Lie conformal algebra, and $T\colon \A\rightarrow\A$ a Rota-Baxter operator of weight $\a\in\C$. Then  $(\A,\{\cdot_\la\cdot\})$ is a Lie conformal algebra by the following $\la$-bracket (cf. Proposition \ref{descendent})
	\begin{align*}
		[x_\lambda y]^T = [Tx_\lambda y]+[x_\lambda Ty]+\a [x_\lambda y], ~~ \forall~x,y\in\A.
	\end{align*}  
	However, $(\A,[\cdot_\la\cdot])$ and $(\A,\{\cdot_\la \cdot\})$ need not be compatible. For example, let
	$\A=\mathrm{Cur}(\mathfrak{sl}_2)$ with basis $\{e,h,f\}$ and nonzero brackets
	\[
		[h_\lambda e]=2e,\qquad [h_\lambda f]=-2f,\qquad [e_\lambda f]=h.
	\]
	Let $T(e)=e$ and $T(h)=T(f)=0$, extended $\C[\p]$-linearly. Since
	$\C e$ and $\C h\oplus \C f$ are Lie subalgebras of $\mathfrak{sl}_2$, $T$ is a Rota-Baxter operator of weight $-1$. The descendant bracket is determined by
	\[
		[h_\lambda f]^T=2f,\qquad [h_\lambda e]^T=[e_\lambda f]^T=0.
	\]
	If the two brackets were compatible, then their sum would define a Lie conformal algebra. But for
	$[\cdot_\lambda\cdot]'=[\cdot_\lambda\cdot]+[\cdot_\lambda\cdot]^T$, one has
	\[
		[h_\lambda [e_\mu f]']'-[[h_\lambda e]'_{\lambda+\mu}f]'-[e_\mu [h_\lambda f]']'
		=[h_\lambda h]'-[2e_{\lambda+\mu}f]'-0=-2h\neq 0.
	\]
	Hence the Jacobi identity fails, and the original bracket and the descendant bracket are not compatible in this example.
\end{ex}

\begin{defi}\label{Compatible Re}\rm
	A representation of a compatible Lie conformal algebra $(\mathcal{A},[\cdot _\la \cdot ]^1,[\cdot _\la \cdot]^2) $ is a triple $(M,\rho,\sigma)$, such that $(M,\rho)$, $(M,\sigma)$ and  $(M,\rho +\sigma)$ are representations of $(\A,[\cdot_\la \cdot]^1)$, $(\A,[\cdot_\la \cdot]^2)$ and $(\A, [\cdot _\la \cdot]^1 + [\cdot_\la \cdot]^2)$, respectively. 
\end{defi}

\begin{theo}\label{prop3.7}
	Let $(\mathcal{A},\left [\cdot _{\lambda}\cdot\right ]^1, [\cdot _\la \cdot]^2) $ be a compatible LCA, let $M$ be a finite $ \mathbb{C}\left [ \partial\right ] $-module and let  $\rho,\sigma:\mathcal{A}\rightarrow \gc\left(M\right)$ be a pair of $\mathbb{C}$-linear maps. Then the following conditions are equivalent:
	\begin{itemize}[(iii)]
		\item[(i)] $\left(M,\rho,\sigma\right)$ is a representation of $ (\mathcal{A},[ \cdot _{\lambda}\cdot]^1,[\cdot _{\lambda}\cdot ]^2)$.
		\item[(ii)]  For any $x,y\in \mathcal{A}$, the following equations hold:
		\begin{align}
			\rho(\p(x))_\la &= -\la\rho(x)_\la,\label{representation1'}\\
			\sigma(\p(x))_\la &= -\la\sigma(x)_\la,\label{representation2'}\\
			\rho( {\left [ x_\lambda y \right ]}^1 )_{\lambda+\mu}&=\rho\left ( x \right )_\lambda\rho\left ( y \right )_\mu - \rho\left ( y \right )_\mu\rho\left ( x \right )_\lambda,\label{representation1}\\
			\sigma(  [x_\lambda y ]^2 )_{\lambda+\mu}&=\sigma\left ( x \right )_\lambda \sigma\left ( y \right )_\mu - \sigma\left ( y \right )_\mu\sigma\left ( x \right )_\lambda,\label{representation2}\\
			\rho( [x_\lambda y ]^2 )_{\lambda+\mu}+\sigma ( [  x_\lambda y ]^1)_{\lambda+\mu}&=\rho\left ( x \right )_\lambda \sigma\left ( y \right )_\mu - \sigma\left ( y \right )_\mu \rho\left ( x \right )_\lambda + \sigma\left ( x \right )_\lambda \rho\left ( y \right )_\mu 
			- \rho\left ( y \right )_\mu\sigma\left ( x \right )_\lambda.\label{representation12}	      
		\end{align}
		\item[(iii)] 
		There is a compatible LCA structure on $\mathcal{A}\oplus M$ defined by
		\begin{align}
			\left [\!\![  \left (x+u  \right )_\lambda\left (  y+v\right )\right]\!\!]^{1}&=\left [ x_\lambda y \right ]^1+\rho(x)_\lambda v-\rho\left ( y \right )_{-\partial-\lambda} u,\\
			\left [\!\![  \left (x+u  \right )_\lambda\left (  y+v\right )\right ]\!\!]^{2}&=[ x_\lambda y ]^2+\sigma(x)_\lambda v-\sigma\left ( y \right )_{-\partial-\lambda}u, 
		\end{align}
		for all $x,y\in \mathcal{A}$ and $u,v\in M$. 
	The resulting  compatible LCA is called the semidirect product of  $(\mathcal{A},[\cdot_\la \cdot]^1,[\cdot_\la \cdot]^2)$ and 
		$(M,\rho,\sigma)$,  and denoted by $\tilde{\A} \ltimes _{\rho,\sigma} M $. 
	\end{itemize}
\end{theo}
\begin{proof}
	First, we show that (i) and (ii) are equivalent. \eqref{representation1'} and \eqref{representation1} hold if and only if $(M,\rho)$ is a representation of $(\A,[\cdot_\la \cdot]^1)$; \eqref{representation2'} and \eqref{representation2} hold if and only if $(M,\sigma)$ is a representation of $(\A,[\cdot_\la \cdot]^2)$.
		If $\left(M,\rho,\sigma\right)$ is a representation of $ (\mathcal{A},[ \cdot _{\lambda}\cdot]^1,[ \cdot _{\lambda}\cdot ]^2) $, then $(M,\rho + \sigma)$ is a representation of $(\A, [\cdot_\la \cdot]^1+[\cdot _\la \cdot]^2)$, which implies that \eqref{representation12} holds. Conversely, for any $x,y\in \mathcal{A} $, by \eqref{representation1'}-\eqref{representation12}, we have
	\begin{align*}
		&(\rho + \sigma )( [ x_{\lambda} y ]^1+ [ x_{\lambda} y ]^2)_{\lambda+\mu}\\
		&=\rho([x_{\lambda} y]^1)_{\lambda+\mu}+ \sigma([x_\lambda y]^1)_{\lambda+\mu}+\rho([x_\lambda y]^2)_{\lambda+\mu} + \sigma([x_\lambda y]^2)_{\lambda+\mu}\\
		&=\rho(x)_\lambda \rho(y)_\mu - \rho(y)_\mu \rho(x)_\lambda +\rho(x)_\lambda \sigma(y)_\mu - \sigma(y)_\mu\rho(x)_\lambda + \sigma(x)_\lambda\rho(y)_\mu \\
		&\quad - \rho(y)_\mu \sigma(x)_\lambda
		+ \sigma(x)_\lambda \sigma(y)_\mu - \sigma(y)_\mu \sigma(x)_\lambda\\
		&=(\rho + \sigma)(x)_\lambda (\rho+\sigma)(y)_\mu - (\rho + \sigma)(y)_\mu (\rho + \sigma)(x)_\lambda,
	\end{align*}
	and
	\begin{align*}
		(\rho+\sigma)(\p(x))_\la = \rho(\p(x))_\la + \sigma(\p(x))_\la = -\la \rho(x)_\la - \la \sigma(x)_\la= -\la(\rho+\sigma)(x)_\la.
	\end{align*}    
	Hence $(M,\rho + \sigma)$ is a representation of the Lie conformal algebra $ (\mathcal{A}, [ \cdot_{\lambda}\cdot]^1+ [\cdot_{\lambda}\cdot]^2)$. Therefore, $\left(M,\rho,\sigma\right)$ is a representation of $ (\mathcal{A},[ \cdot _{\lambda}\cdot]^1, [\cdot _{\lambda}\cdot ]^2) $.
	
	Now, we prove that (ii) and (iii) are equivalent. In fact, we have the following correspondences:
	\eqref{representation1} and \eqref{representation1'} hold if and only if $(\mathcal{A} \oplus M,\left[\!\![\cdot_\lambda\cdot]\!\!\right]^1)$ is a Lie conformal algebra;
	\eqref{representation2} and \eqref{representation2'} hold if and only if $(\mathcal{A} \oplus M,\left[\!\![\cdot_\lambda\cdot\right]\!\!]^2)$ is a Lie conformal algebra;
	\eqref{representation12} holds if and only if the compatible Jacobi identity of $(\mathcal{A} \oplus M,\left[\!\![\cdot_\lambda\cdot]\!\!\right]^1)$ and $(\mathcal{A} \oplus M,\left[\!\![\cdot_\lambda\cdot]\!\!\right]^2)$ holds.
	\end{proof}

\begin{prop}
	Let $M$ be a finite $\C[\partial]$-module. Let $(M,\rho,\sigma)$ be a representation of a compatible LCA $(\A,[\cdot_\la\cdot]^1,[\cdot_\la\cdot]^2)$, and let $(M^{*c},\rho^*)$ and $(M^{*c},\sigma^*)$ be the dual
	representation of $(M,\rho)$ and $(M,\sigma)$, respectively.  Then $(M^{*c},\rho^*,\sigma^*)$ is a representation of $(\A,[\cdot_\la\cdot]^1,[\cdot_\la\cdot]^2)$, which is called the dual representation of $(M,\rho,\sigma)$.  In particular, $(\A^{*c},\ad^*_{1},\ad^*_{2})$ is a dual  representation of the adjoint representation $(\A, \ad_{1},\ad_{2})$ of $(\A,[\cdot_\la\cdot]^1,[\cdot_\la\cdot]^2)$, where the adjoint actions are defined by $\ad_1(a)_\la b = [a_\la b]^1 $, and  $\ad_2(a)_\la b = [a_\la b]^2$, for all $a,b\in \A$.
\end{prop}
\begin{proof}
Assume that $(M,\rho,\sigma)$ is a representation of the compatible Lie conformal algebra
$(\A,[\cdot_{\la}\cdot]^1,[\cdot_{\la}\cdot]^2)$.  By Definition \ref{Compatible Re} this means that
$(M,\rho)$ is a representation of $(\A,[\cdot_{\la}\cdot]^1)$,
$(M,\sigma)$ is a representation of $(\A,[\cdot_{\la}\cdot]^2)$, and
$(M,\rho+\sigma)$ is a representation of $(\A,[\cdot_{\la}\cdot]^1+[\cdot_{\la}\cdot]^2)$.
Dual representation theory for Lie conformal algebras immediately yields that
$(M^{*c},\rho^{*})$ and $(M^{*c},\sigma^{*})$ are representations of
$(\A,[\cdot_{\la}\cdot]^1)$ and $(\A,[\cdot_{\la}\cdot]^2)$ respectively.

It remains to verify that $(M^{*c},\rho^{*}+\sigma^{*})$ is a representation of
$(\A,[\cdot_{\la}\cdot]^1+[\cdot_{\la}\cdot]^2)$.  We first show that duality is linear:
$(\rho+\sigma)^{*}=\rho^{*}+\sigma^{*}$.
Indeed, for any $a\in\A$, $\varphi\in M^{*c}$, $v\in M$,
\[
\begin{aligned}
\big((\rho+\sigma)^{*}(a)_{\la}\varphi\big)_{\mu}v
&=-\varphi_{\mu-\la}\big((\rho+\sigma)(a)_{\la}v\big)\\
&=-\varphi_{\mu-\la}\big(\rho(a)_{\la}v+\sigma(a)_{\la}v\big)\\
&=-\varphi_{\mu-\la}(\rho(a)_{\la}v)-\varphi_{\mu-\la}(\sigma(a)_{\la}v)\\
&=(\rho^{*}(a)_{\la}\varphi)_{\mu}v+(\sigma^{*}(a)_{\la}\varphi)_{\mu}v.
\end{aligned}
\]
Since $(M,\rho+\sigma)$ is a representation of $(\A,[\cdot_\la\cdot]^1+[\cdot_\la\cdot]^2)$, its dual
$(M^{*c},(\rho+\sigma)^{*})$ is again a representation of the same algebra. The equality
$(\rho+\sigma)^{*}=\rho^{*}+\sigma^{*}$ shows that $(M^{*c},\rho^{*}+\sigma^{*})$
 is a representation of $(\A,[\cdot_{\la}\cdot]^1+[\cdot_{\la}\cdot]^2)$.
Therefore,
$(M^{*c},\rho^{*},\sigma^{*})$ is a representation of the compatible
Lie conformal algebra $(\A,[\cdot_{\la}\cdot]^1,[\cdot_{\la}\cdot]^2)$.

Finally, taking the adjoint representation $\ad_1(a)_{\la}b=[a_{\la}b]^1$,
$\ad_2(a)_{\la}b=[a_{\la}b]^2$ on $\A$, one immediately checks that
$(\A,\ad_1,\ad_2)$ is a representation of the compatible Lie conformal
algebra.  Applying the statement just proved, its dual
$(\A^{*c},\ad_1^{*},\ad_2^{*})$ becomes a dual representation.
\end{proof}

\subsection{Matched pairs of compatible Lie conformal algebras}

In this subsection, all Lie conformal algebra are assumed to be finitely generated as $\C[\p]$-modules. 
Let $(\A,[\cdot_\la\cdot])$ and $(\mathcal{B},\{\cdot_\lambda\cdot\})$ be two Lie conformal algebras. Suppose that \( \rho : \A \rightarrow Cend(\B) \) and \( \sigma : \B \rightarrow Cend(\A) \) are two representations, satisfying the following conditions:
\begin{align}\label{matched pair 1}
	&\rho(x)_{\lambda}\{a_{\mu}b\} - \{(\rho(x)_{\lambda}a)_{\lambda+\mu}b\} - \{a_{\mu}(\rho(x)_{\lambda}b)\} 
	+ \rho(\sigma(a)_{-\lambda-\partial }x)_{\lambda+\mu}b - \rho(\sigma(b)_{-\lambda-\partial }x)_{-\mu-\partial }a = 0,\\
	\label{matched pair 2}
	& \sigma(a)_{-\lambda-\mu-\partial}[x_\lambda y] - [x_{\lambda}(\sigma(a)_{-\mu-\partial }y)] + [y_{\mu}(\sigma(a)_{-\lambda-\partial }x)]
	+ \sigma(\rho(x)_{\lambda}a)_{-\mu-\partial }y - \sigma(\rho(y)_{\mu}a)_{-\lambda-\partial }x = 0,
\end{align}
for any \( x, y \in \A \) and \( a, b \in \B \). Then there is a Lie conformal algebra structure on \( \A \oplus \B \) via
\begin{align}
	[\!\![(x+a)_{\lambda}(y+b)]\!\!] &= \left( [x_{\lambda}y] + \sigma(a)_{\lambda}y - \sigma(b)_{-\lambda-\partial }x \right) + \left( \{a_{\lambda}b\} + \rho(x)_{\lambda}b - \rho(y)_{-\lambda-\partial }a \right),\label{bowtie}
\end{align}
for any \( x, y \in \A \) and \( a, b \in \B \). We denote this Lie conformal algebra by \( \A \bowtie \B \). The quadruple \((\A, \B;\rho, \sigma)\) satisfying the above conditions is called a {\it matched pair of Lie conformal algebras} (see \cite{HL3}).

Now we define matched pairs of compatible Lie conformal algebras. 

\begin{defi}\label{CMPD}{\rm
		Let $(\mathcal{A},[\cdot_\lambda\cdot]^{1},[\cdot_\lambda\cdot]^{2})$ and  $(\mathcal{B},\{\cdot_\lambda\cdot\}^{1},\{\cdot_\lambda\cdot\}^{2})$ be two compatible LCAs. Let $\rho_{\mathcal{A}},\sigma_{\mathcal{A}}: \mathcal{A} \rightarrow Cend(\mathcal{B})$ and $\rho_{\mathcal{B}},\sigma_{\mathcal{B}}$ : $\mathcal{B} \rightarrow Cend(\mathcal{A})$ be four $\C$-linear maps. The 6-tuple 
		$(\mathcal{A},\mathcal{B};\rho_{\mathcal{A}},\sigma_{\mathcal{A}},\rho_{\mathcal{B}},\sigma_{\mathcal{B}})$ is called a  matched pair of compatible LCAs, if $\big((\A,[\cdot_\la\cdot]^1),(\B,\{\cdot_\lambda\cdot\}^{1});\rho_{\A},\rho_{\B}\big)$,   $\big((\A,[\cdot_\la\cdot]^2),(\B,\{\cdot_\lambda\cdot\}^{2});\sigma_{\A},\sigma_{\B}\big)$ and $\big((\A,[\cdot_\la \cdot]^1 + [\cdot_\la \cdot]^2),(\B,\{\cdot_\lambda\cdot\}^{1} + \{\cdot_\lambda\cdot\}^{2});\rho_{\A}+\sigma_{\A},\rho_{\B}+\sigma_{\B}\big)$
		are matched pairs of LCAs. }
\end{defi}

\begin{theo}\label{MP compatible}
	Let $\tilde{\A}=(\mathcal{A},[\cdot_\lambda\cdot]^{1},[\cdot_\lambda\cdot]^{2})$, $\tilde{\B}=(\mathcal{B},\{\cdot_\lambda\cdot\}^{1},\{\cdot_\lambda\cdot\}^{2})$ be two compatible LCAs. 
	Let $\rho_{\mathcal{A}},\sigma_{\mathcal{A}}$ : $\mathcal{A} \rightarrow gc(\mathcal{B})$ and $\rho_{\mathcal{B}},\sigma_{\mathcal{B}}$ : $\mathcal{B} \rightarrow gc(\mathcal{A})$ be four $\C$-linear maps. 
	Then the following conditions are equivalent:
	\begin{itemize}[(iii)]
		\item[(i)]$(\mathcal{A},\mathcal{B};\rho_{\mathcal{A}},\sigma_{\mathcal{A}},\rho_{\mathcal{B}},\sigma_{\mathcal{B}})$ is a matched pair of compatible LCAs.
		\item [(ii)]
				$(\B,\rho_{\mathcal{A}},\sigma_{\mathcal{A}})$
			 is a representation of $\tilde{\A}$ and $(\A,\rho_{\mathcal{B}},\sigma_{\mathcal{B}})$is a representation of $\tilde{\B}$, and the following equations hold for all $x,y\in \mathcal{A}$, $a,b\in \mathcal{B}$:
				\begin{align}
					&\rho_{\A}(x)_{\lambda}\{a_\mu b\}^{1}
					-\{(\rho_{\A}(x)_{\lambda}a)_{\lambda+\mu}b\}^{1}
					-\{a_\mu(\rho_{\A}(x)_{\lambda}b)\}^{1}\nonumber\\
					&\quad+\rho_{\A}(\rho_{\B}(a)_{-\lambda-\partial}x)_{\lambda+\mu}b
					-\rho_{\A}(\rho_{\B}(b)_{-\lambda-\partial}x)_{-\mu-\partial}a=0;\label{mp1}\\[0.25em]
					&\rho_{\B}(a)_{-\lambda-\mu-\partial}[x_\lambda y]^1
					-[x_{\lambda}(\rho_{\B}(a)_{-\mu-\partial}y)]^1
					+[y_{\mu}(\rho_{\B}(a)_{-\lambda-\partial}x)]^1\nonumber\\
					&\quad+\rho_{\B}(\rho_{\A}(x)_{\lambda}a)_{-\mu-\partial}y
					-\rho_{\B}(\rho_{\A}(y)_{\mu}a)_{-\lambda-\partial}x=0;\label{mp2}\\[0.25em]
					&\sigma_{\A}(x)_{\lambda}\{a_\mu b\}^{2}
					-\{(\sigma_{\A}(x)_{\lambda}a)_{\lambda+\mu}b\}^{2}
					-\{a_\mu(\sigma_{\A}(x)_{\lambda}b)\}^{2}\nonumber\\
					&\quad+\sigma_{\A}(\sigma_{\B}(a)_{-\lambda-\partial}x)_{\lambda+\mu}b
					-\sigma_{\A}(\sigma_{\B}(b)_{-\lambda-\partial}x)_{-\mu-\partial}a=0;\label{mp3}\\[0.25em]
					&\sigma_{\B}(a)_{-\lambda-\mu-\partial}[x_\lambda y]^2
					-[x_{\lambda}(\sigma_{\B}(a)_{-\mu-\partial}y)]^2
					+[y_{\mu}(\sigma_{\B}(a)_{-\lambda-\partial}x)]^2\nonumber\\
					&\quad+\sigma_{\B}(\sigma_{\A}(x)_{\lambda}a)_{-\mu-\partial}y
					-\sigma_{\B}(\sigma_{\A}(y)_{\mu}a)_{-\lambda-\partial}x=0;\label{mp4}\\[0.25em]
					&\sigma_{\A}(x)_{\lambda}\{a_\mu b\}^1
					-\{a_\mu(\sigma_{\A}(x)_{\lambda}b)\}^{1}
					-\{(\sigma_{\A}(x)_{\lambda}a)_{\lambda+\mu}b\}^{1}
					+\sigma_{\A}(\rho_{\B}(a)_{-\lambda-\partial}x)_{\lambda+\mu}b\nonumber\\
					&\quad-\sigma_{\A}(\rho_{\B}(b)_{-\lambda-\partial}x)_{-\mu-\partial}a
					+\rho_{\A}(x)_{\lambda}\{a_\mu b\}^{2}
					-\{a_\mu(\rho_{\A}(x)_{\lambda}b)\}^{2}
					-\{(\rho_{\A}(x)_{\lambda}a)_{\lambda+\mu}b\}^{2}\nonumber\\
					&\quad+\rho_{\A}(\sigma_{\B}(a)_{-\lambda-\partial}x)_{\lambda+\mu}b
					-\rho_{\A}(\sigma_{\B}(b)_{-\lambda-\partial}x)_{-\mu-\partial}a=0;\label{compatible matched pair1}\\[0.25em]
					&\sigma_{\B}(a)_{-\lambda-\mu-\partial}[x_\lambda y]^{1}
					-[x_\lambda(\sigma_{\B}(a)_{-\mu-\partial}y)]^{1}
					+[y_\mu(\sigma_{\B}(a)_{-\lambda-\partial}x)]^{1}
					+\sigma_{\B}(\rho_{\A}(x)_\lambda a)_{-\mu-\partial}y\nonumber\\
					&\quad-\sigma_{\B}(\rho_{\A}(y)_\mu a)_{-\lambda-\partial}x
					+\rho_{\B}(a)_{-\lambda-\mu-\partial}[x_\lambda y]^{2}
					-[x_\lambda(\rho_{\B}(a)_{-\mu-\partial}y)]^{2}
					+[y_\mu(\rho_{\B}(a)_{-\lambda-\partial}x)]^{2}\nonumber\\
					&\quad+\rho_{\B}(\sigma_{\A}(x)_\lambda a)_{-\mu-\partial}y
					-\rho_{\B}(\sigma_{\A}(y)_\mu a)_{-\lambda-\partial}x=0.\label{compatible matched pair2}
				\end{align}
		\item[(iii)]  There is a compatible Lie conformal algebra structure on $\mathcal{A} \oplus \mathcal{B}$ defined by
		\begin{equation}
			\begin{aligned}
				\relax {[\!\![{(x+a)}_{\lambda} {(y+b)} ]\!\!]}^1 &= [x_\lambda y]^{1} 
				+ \rho_{\mathcal{B}}(a)_\lambda y 
				- \rho_{\mathcal{B}}(b)_{-\lambda - \partial}x 
				+ \{a_\lambda b\}^{1} 
				+ \rho_{\mathcal{A}}(x)_\lambda b 
				- \rho_{\mathcal{A}}(y)_{-\lambda - \partial}a,\label{A bowtie B 1}
			\end{aligned}
		\end{equation}
		\begin{equation}
			\begin{aligned}
				\relax {[\!\![ {(x+a)}_{\lambda} {(y+b)} ]\!\!]}^2 &= [x_\lambda y]^{2}
				+ \sigma_{\mathcal{B}}(a)_\lambda y 
				- \sigma_{\mathcal{B}}(b)_{-\lambda - \partial}x 
				+ \{a_\lambda b\}^{2} 
				+ \sigma_{\mathcal{A}}(x)_\lambda b 
				- \sigma_{\mathcal{A}}(y)_{-\lambda - \partial}a,\label{A bowtie B 2}
			\end{aligned}
		\end{equation}
		for any $x,y\in \mathcal{A}$, $a,b\in \mathcal{B}$.  We denote this compatible Lie conformal algebra by $\tilde{\A}\bowtie_{\rho_{\mathcal{B}},\sigma_{\mathcal{B}}}^{\rho_{\mathcal{A}},\sigma_{\mathcal{A}}}\tilde{\B} $ or simply $\tilde{\A}\bowtie \tilde{\B}$.
	\end{itemize} 
\end{theo}

\begin{proof}
	First, we show (i) and (ii) are equivalent. 
	\eqref{mp1} and \eqref{mp2} hold if and only if $((\A,[\cdot_\la\cdot]^1),(\B,\{\cdot_\lambda\cdot\}^{1});\rho_{\A},\rho_{\B})$ is a matched pair of LCAs, while \eqref{mp3} and \eqref{mp4} hold if and only if $((\A,[\cdot_\la\cdot]^2),(\B,\{\cdot_\lambda\cdot\}^{2});\sigma_{\A},\sigma_{\B})$ is a matched pair of LCAs. 
	
	Let $(\mathcal{A},\mathcal{B};\rho_{\mathcal{A}},\sigma_{\mathcal{A}},\rho_{\B},\sigma_{\B})$ be a matched pair of compatible LCAs. By Definition \ref{CMPD},
	$(\B,\rho_{\mathcal{A}})$, $(\B,\sigma_{\mathcal{A}})$ and $(\B,\rho_{\mathcal{A}}+\sigma_{\mathcal{A}})$ are representations of
	$(\A,[\cdot_\lambda\cdot]^1)$, $(\A,[\cdot_\lambda\cdot]^2)$ and
	$(\A,[\cdot_\lambda\cdot]^1+[\cdot_\lambda\cdot]^2)$, respectively. Hence
	$(\B,\rho_{\mathcal{A}},\sigma_{\mathcal{A}})$ is a representation of $(\mathcal{A},[\cdot_\lambda\cdot]^{1},[\cdot_\lambda\cdot]^{2})$. The same argument gives that
	$(\A,\rho_{\mathcal{B}},\sigma_{\mathcal{B}})$ is a representation of $(\mathcal{B},\{\cdot_\lambda\cdot\}^{1},\{\cdot_\lambda\cdot\}^{2})$.  We also have  $((\A,[\cdot_\la \cdot]^1 + [\cdot_\la \cdot]^2),(\B,\{\cdot_\lambda\cdot\}^{1} + \{\cdot_\lambda\cdot\}^{2});\rho_{\A}+\sigma_{\A},\rho_{\B}+\sigma_{\B})$
	is a matched pair of LCAs. This implies that \eqref{compatible matched pair1} and \eqref{compatible matched pair2} hold.
	
	Conversely, suppose that \eqref{mp1}-\eqref{compatible matched pair2} holds. For all $x,y\in \mathcal{A}$, $a,b\in \mathcal{B}$, we have 
	\begin{align*}
		&(\rho_{\mathcal{A}}+\sigma_{\mathcal{A}})(x)_\lambda(\{a_\mu b\}^{1}+\{a_\mu b\}^{2})\\
		&= \rho_{\mathcal{A}}(x)_\lambda\{a_\mu b\}^{1} + \sigma_{\mathcal{A}}(x)_\lambda\{a_\mu b\}^{2} 
		+ \rho_{\mathcal{A}}(x)_\lambda\{a_\mu b\}^{2} + \sigma_{\mathcal{A}}(x)_\lambda\{a_\mu b\}^{1} \\
		&= \{(\rho_{\mathcal{A}}(x)_{\lambda}a)_{\lambda+\mu}b\}^{1} + \{a_{\mu}(\rho_{\mathcal{A}}(x)_{\lambda}b)\}^{1} 
		- \rho_{\mathcal{A}}(\rho_{\mathcal{B}}(a)_{-\lambda-\partial}x)_{\lambda+\mu}b  
		+ \rho_{\mathcal{A}}(\rho_{\mathcal{B}}(b)_{-\lambda-\partial}x)_{-\mu-\partial}a 
		+ \{(\sigma_{\mathcal{A}}(x)_{\lambda}a)_{\lambda+\mu}b\}^{2} \\
		&\quad+ \{a_{\mu}(\sigma_{\mathcal{A}}(x)_{\lambda}b)\}^{2} 
		- \sigma_{\mathcal{A}}(\sigma_{\mathcal{B}}(a)_{-\lambda-\partial}x)_{\lambda+\mu}b + \sigma_{\mathcal{A}}(\sigma_{\mathcal{B}}(b)_{-\lambda-\partial}x)_{-\mu-\partial}a + \{a_\mu (\sigma_{\mathcal{A}}(x)_\lambda b)\}^{1} 
		+\{(\sigma_{\mathcal{A}}(x)_\lambda a)_{\lambda + \mu}b\}^{1} \\
		&\quad - \sigma_{\mathcal{A}}(\rho_{\mathcal{B}}(a)_{-\lambda - \partial}x)_{\lambda + \mu}b + \sigma_{\mathcal{A}}(\rho_{\mathcal{B}}(b)_{-\lambda - \partial}x)_{-\mu - \partial}a
		+ \{a_\mu (\rho_{\mathcal{A}}(x)_\lambda b)\}^{2} 
		+ \{(\rho_{\mathcal{A}}(x)_\lambda a)_{\lambda + \mu}b\}^{2}\\ 
		&\quad - \rho_{\mathcal{A}}(\sigma_{\mathcal{B}}(a)_{-\lambda - \partial}x)_{\lambda + \mu}b
		+ \rho_{\mathcal{A}}(\sigma_{\mathcal{B}}(b)_{-\lambda - \partial}x)_{-\mu - \partial}a \\
		&=\{((\rho_{\mathcal{A}}+\sigma_{\mathcal{A}})(x)_{\lambda}a)_{\lambda+\mu}b\}^{1} 
		+ \{((\rho_{\mathcal{A}}+\sigma_{\mathcal{A}})(x)_{\lambda}a)_{\lambda+\mu}b\}^{2} 
		+ \{a_{\mu}((\rho_{\mathcal{A}}+\sigma_{\mathcal{A}})(x)_{\lambda}b)\}^{1} 
		+ \{a_{\mu}((\rho_{\mathcal{A}}+\sigma_{\mathcal{A}})(x)_{\lambda}b)\}^{2} \\
		&\quad - (\rho_{\mathcal{A}}+\sigma_{\mathcal{A}})((\rho_{\mathcal{B}}+\sigma_{\mathcal{B}})(a)_{-\lambda-\partial}x)_{\lambda+\mu}b
		+ (\rho_{\mathcal{A}}+\sigma_{\mathcal{A}})((\rho_{\mathcal{B}}+\sigma_{\mathcal{B}})(b)_{-\lambda-\partial}x)_{-\mu-\partial}a.
	\end{align*}
	This is exactly \eqref{matched pair 1} in the case of the matched pair of LCAs $(\A,[\cdot_\la \cdot]^1 + [\cdot_\la \cdot]^2)$ and $(\B,\{\cdot_\lambda\cdot\}^{1} + \{\cdot_\lambda\cdot\}^{2})$, along with the map  $\rho_{\A}+\sigma_{\A}$. By a similar calculation, one can show \eqref{matched pair 2} holds for 
	$[\cdot_\lambda\cdot]^1+[\cdot_\lambda\cdot]^2$, $\{\cdot_\lambda\cdot\}^1+\{\cdot_\lambda\cdot\}^2$ and
	$\rho_{\B}+\sigma_{\B}$.
	Therefore,  $((\A,[\cdot_\la \cdot]^1 + [\cdot_\la \cdot]^2),(\B,\{\cdot_\lambda\cdot\}^{1} + \{\cdot_\lambda\cdot\}^{2});\rho_{\A}+\sigma_{\A},\rho_{\B}+\sigma_{\B})$
	is a matched pair of LCAs, and so  $(\mathcal{A},\mathcal{B};\rho_{\mathcal{A}},\sigma_{\mathcal{A}},\rho_{\mathcal{B}},\sigma_{\mathcal{B}})$ is a matched pair of compatible LCAs.

	Next, we prove that (ii) and (iii) are equivalent. We see that
	$(\mathcal{A} \oplus \mathcal{B},[\!\![\cdot_\lambda\cdot]\!\!]^1)$ is a Lie conformal algebra if and only if \eqref{mp1} and \eqref{mp2} hold, and
	$(\mathcal{A} \oplus \mathcal{B},[\!\![\cdot_\lambda\cdot]\!\!]^2)$ is a Lie conformal algebra if and only if \eqref{mp3} and \eqref{mp4} hold. Thus
	$(\mathcal{A} \oplus \mathcal{B},[\!\![\cdot_\lambda\cdot]\!\!]^1,[\!\![\cdot_\lambda\cdot]\!\!]^2)$ is a compatible Lie conformal algebra if and only if the compatible Jacobi identity holds. Let
	\[
	\begin{aligned}
	R(u,v,w)=\ [\!\![ u_\lambda {[\!\![ v_\mu w ]\!\!]}^1 ]\!\!]^2
	+ [\!\![ u_\lambda {[\!\![ v_\mu w ]\!\!]}^2 ]\!\!]^1
	- [\!\![ {[\!\![ u_\lambda v ]\!\!]}^{1}_{\lambda+\mu} w ]\!\!]^2
	- [\!\![ v_\mu {[\!\![ u_\lambda w ]\!\!]}^1 ]\!\!]^2
	- [\!\![ {[\!\![ u_\lambda v ]\!\!]}^{2}_{\lambda+\mu} w ]\!\!]^1
	- [\!\![ v_\mu {[\!\![ u_\lambda w ]\!\!]}^2 ]\!\!]^1,
	\end{aligned}
	\]
	where $u,v,w\in \mathcal{A}\oplus\mathcal{B}$. We need to examine
	\[
	R(x,y,z)=0,\quad R(a,b,c)=0,\quad R(a,b,x)=0,\quad R(x,y,a)=0,
	\]
	where $x,y,z\in \mathcal A$ and $a,b,c\in \mathcal B$. In fact,
	$R(x,y,z)=0$ holds if and only if
	$(\mathcal{A},[\cdot_\lambda\cdot]^1,[\cdot_\lambda\cdot]^2)$ is a compatible Lie conformal algebra;
	$R(a,b,c)=0$ holds if and only if
	$(\mathcal{B},\{\cdot_\lambda\cdot\}^{1},\{\cdot_\lambda\cdot\}^{2})$ is a compatible Lie conformal algebra;
	$R(a,b,x)=0$ holds if and only if
	$(\mathcal A,\rho_{\mathcal B},\sigma_{\mathcal B})$ is a representation of
	$(\mathcal B,\{\cdot_\lambda\cdot\}^{1},\{\cdot_\lambda\cdot\}^{2})$ and \eqref{compatible matched pair1} holds;
	$R(x,y,a)=0$ holds if and only if
	$(\mathcal B,\rho_{\mathcal A},\sigma_{\mathcal A})$ is a representation of
	$(\mathcal A,[\cdot_\lambda\cdot]^1,[\cdot_\lambda\cdot]^2)$ and \eqref{compatible matched pair2} holds. The proof is finished.
\end{proof}


\section{Compatible Lie conformal coalgebras}
\vs{8pt}

In this section, we define compatible Lie conformal coalgebras and provide constructions of compatible Lie conformal algebras from finite compatible Lie conformal coalgebras and vice versa. These results extend the duality theory of Lie conformal algebras to the compatible case. 

\begin{defi}(\cite{Li}){\rm
A  Lie conformal coalgebra $\A$ is a $\mathbb{C}[\partial]$-module endowed with a $\mathbb{C}[\partial]$-homomorphism
\[
\delta:\A\to  \wedge^2 \A
\]
such that
\begin{align}
(\id\otimes\delta)\delta-\tau_{12}(\id\otimes\delta)\delta=(\delta\otimes \id)\delta,
\end{align}
where $\tau_{12}(a\otimes b\otimes c)=b\otimes a\otimes c$, for $a,b,c\in\A$.}
\end{defi}
\begin{defi}(\cite{Li}){\rm
A Lie conformal bialgebra is a triple $(\A,[\cdot_\la \cdot],\delta)$, where $(\A,[\cdot_\la \cdot])$ is a Lie conformal algebra, $(\A,\delta)$ is a Lie conformal coalgebra and they satisfy the cocycle condition:
\begin{align}
&\de ([x_\la y]) = \left(\ad(x)_\la \otimes \id + \id \otimes \ad(x)_\la\right)\de(y)
									-\left(\ad(y)_{-\la -\p}\otimes \id + \id \otimes \ad(y)_{-\la -\p}\right)\de(x), \; \forall x,y \in \A.
\end{align}}
\end{defi}

\begin{prop}\label{bialgebra prop1}(\cite{Li})
(a) Let $(\A, \delta)$ be a finite Lie conformal coalgebra, then $\A^{*c} = \mathrm{Chom}(\A, \mathbb{C})$ is a Lie conformal algebra with the following $\la$-bracket $(f, g \in \A^{*c})$:
\begin{align}
([f_{\mu}g])_{\lambda}(r) = \sum f_{\mu}(r_{(1)})g_{\lambda - \mu}(r_{(2)}) = (f \otimes g)_{\mu, \lambda - \mu}(\delta(r)),
\end{align}
where $\delta(r) = \sum r_{(1)} \otimes r_{(2)}$.

(b) Let $(\A, [\cdot_{\lambda}\cdot])$ be a Lie conformal algebra free of finite rank, that is $\A = \bigoplus_{i=1}^{n} \mathbb{C}[\partial]a^i$, then $\A^{*c} = \mathrm{Chom}(\A, \mathbb{C}) = \bigoplus_{i=1}^{n} \mathbb{C}[\partial]a_i$, where $\{a_i\}$ is a dual $\mathbb{C}[\partial]$-basis in the sense that $(a_i)_{\lambda}(a^j) = \delta_{ij}$, is a Lie conformal coalgebra with the following co-bracket:
\begin{align}
\delta(f) = \sum_{i,j}f_{\mu}([a^i_{\lambda}a^j])(a_i \otimes a_j)|_{\lambda = \partial \otimes 1, \mu = -\partial \otimes 1 - 1 \otimes \partial}.
\end{align}
More precisely, if
\[
[a^i_{\lambda}a^j] = \sum_k P_k^{ij}(\lambda, \partial)a^k,
\]
where $P_k^{ij}$ are some polynomials in $\lambda$ and $\partial$, then the co-bracket is
\[
\delta(a_k) = \sum_{i,j} Q_k^{ij}(\partial \otimes 1, 1 \otimes \partial)a_i \otimes a_j,
\]
where $Q_k^{ij}(x, y) = P_k^{ij}(x, -x - y)$.
\end{prop}

We give the definition of compatible Lie conformal coalgebras.

\begin{defi}\label{com Co}{\rm Two Lie conformal coalgebras  $(\A,\de_1)$ and  $(\A,\de_2)$ over the same $\C[\p]$-module $\A$ are called compatible if for all $k_1,k_2\in\C$, $k_1\de_1+k_2\de_2$ defines a Lie conformal coalgebra structure on $\A$. 	}
\end{defi}

If  $(\A,\de_1)$ and  $(\A,\de_2)$ are compatible Lie conformal coalgebras, then we denote it by $(\A,\de_1,\de_2)$.

\begin{remark}
    The condition in the above definition is equivalent to the following compatible co-bracket identity
    \begin{align}
        (\id\otimes \de_1)\de_2 - \tau_{12}(\id\otimes \de_1)\de_2 - (\de_1\otimes \id)\de_2 + (\id\otimes\de_2)\de_1 - \tau_{12}(\id\otimes \de_2)\de_1 - (\de_2\otimes \id)\de_1 = 0.\label{Compatible Co}
    \end{align}
\end{remark}

\begin{ex}\label{ex-Rp-compatible-coalgebra}
	Let $R=\C[\p]a\oplus\C[\p]b$ be the $\C[\p]$-module in Example \ref{ex-Rp-compatible-LCA}. Define two $\C[\p]$-module homomorphisms $\de_1,\de_2:R\to\wedge^2R$ by
	$$
		\de_1(a)=0,\qquad \de_1(b)=\p(a\wedge b),
	$$
	and
	$$
		\de_2(a)=0,\qquad \de_2(b)=\p^3(a\wedge b),
	$$
	where $a\wedge b=a\otimes b-b\otimes a$, and $\p$ acts on $R\otimes R$ as
	$\p\otimes1+1\otimes\p$. The sum cobracket is
	$$
		(\de_1+\de_2)(a)=0,\qquad
		(\de_1+\de_2)(b)=(\p+\p^3)(a\wedge b).
	$$
	By the construction in Example 2.18 of \cite{Li}, the choices $q(\p)=\p$, $q(\p)=\p^3$ and
	$q(\p)=\p+\p^3$ give Lie conformal coalgebras corresponding respectively to $\de_1$, $\de_2$
	and $\de_1+\de_2$. Hence $(R,\de_1,\de_2)$ is a compatible Lie conformal coalgebra.
\end{ex}

The following result is a generalization of Proposition \ref{bialgebra prop1} (a).

\begin{prop}\label{compatible dual}
	Let $(\mathcal{A},\de_1 ,\de_2)$ be a finite compatible Lie conformal coalgebra. Then there is a compatible Lie conformal algebra structure on $\mathcal{A}^{*c}=Chom(\mathcal{A},\mathbb{C})$ with the following $\la$-bracket ($f,g\in \A^{*c},$ $r\in\A$):
	\begin{align}
		([f_{\mu}g]^{1})_{\lambda}(r) = \sum f_{\mu}(r_{(1)})g_{\lambda - \mu}(r_{(2)}) = (f \otimes g)_{\mu, \lambda - \mu}(\de_1(r)),\\
		([f_{\mu}g]^{2})_{\lambda}(r) = \sum f_{\mu}(r_{(3)})g_{\lambda - \mu}(r_{(4)}) = (f \otimes g)_{\mu, \lambda - \mu}(\de_2(r)),
	\end{align}
where $\de_1 (r) = \Sigma r_{(1)}\otimes r_{(2)}$ and $\de_2 (r) = \Sigma r_{(3)}\otimes r_{(4)}$. 
\end{prop}
\begin{proof}
	By Proposition \ref{bialgebra prop1},  $(\mathcal{A}^{*c},[\cdot_\la \cdot]^{1})$ and
	$(\mathcal{A}^{*c},[\cdot_\la \cdot]^{2})$ are Lie conformal algebras.
	Let 
    \begin{align*}
        R(f,g,h)=[f_\lambda [g_\mu h]^{1}]^{2} + [f_\lambda [g_\mu h]^{2}]^{1} 
	-[[f_\lambda g]^{1}_{\lambda+\mu} h]^{2} -[g_\mu [f_\lambda h]^{1}]^{2}
	-[[f_\lambda g]^{2}_{\lambda+\mu} h]^{1} -[g_\mu [f_\lambda h]^{2}]^{1}, ~~\forall~ f,g,h\in\A^{*c}.
    \end{align*}
 It suffices to show $R(f,g,h)=0$. For all $r\in\A$, we have 
	\begin{align}
		&R(f,g,h)_\nu (r)\nonumber\\
		&= [f_\lambda [g_\mu h]^{1}]^{2}{}_\nu(r) + [f_\lambda [g_\mu h]^{2}]^{1}{}_\nu(r) 
			-[[f_\lambda g]^{1}_{\lambda+\mu} h]^{2}{}_\nu(r) -[g_\mu [f_\lambda h]^{1}]^{2}{}_\nu(r)
		 -[[f_\lambda g]^{2}_{\lambda+\mu} h]^{1}{}_\nu(r) -[g_\mu [f_\lambda h]^{2}]^{1}{}_\nu(r)\nonumber\\
		&= \sum (f\otimes g\otimes h)_{\la, \mu, \nu -\la -\mu} \Big( r_{(3)}\otimes (r_{(4)})_{(1)}\otimes (r_{(4)})_{(2)}
		 -(r_{(3)})_{(1)}\otimes (r_{(3)})_{(2)}\otimes r_{(4)} - (r_{(4)})_{(1)}\otimes r_{(3)}\otimes (r_{(4)})_{(2)} \nonumber \\
		 &\quad + r_{(1)}\otimes (r_{(2)})_{(3)}\otimes (r_{(2)})_{(4)}
		 -(r_{(1)})_{(3)}\otimes (r_{(1)})_{(4)}\otimes r_{(2)} - (r_{(2)})_{(3)}\otimes r_{(1)}\otimes (r_{(2)})_{(4)} \Big) .\label{CLCA}
	\end{align}
The tensor in the parentheses is precisely
\[
\big((\id\otimes\de_1)\de_2-\tau_{12}(\id\otimes\de_1)\de_2-(\de_1\otimes\id)\de_2
     +(\id\otimes\de_2)\de_1-\tau_{12}(\id\otimes\de_2)\de_1-(\de_2\otimes\id)\de_1\big)(r).
\]
Indeed, for instance, $r_{(3)}\otimes (r_{(4)})_{(1)}\otimes (r_{(4)})_{(2)}$ is the contribution of
$(\id\otimes\de_1)\de_2(r)$, while
$(r_{(4)})_{(1)}\otimes r_{(3)}\otimes (r_{(4)})_{(2)}$ is the contribution of
$\tau_{12}(\id\otimes\de_1)\de_2(r)$; the remaining four terms are obtained in the same way from
$(\de_1\otimes\id)\de_2(r)$, $(\id\otimes\de_2)\de_1(r)$,
$\tau_{12}(\id\otimes\de_2)\de_1(r)$ and $(\de_2\otimes\id)\de_1(r)$, with the signs displayed in
\eqref{Compatible Co}. Hence the tensor vanishes by \eqref{Compatible Co}, and so
$R(f,g,h)=0$. This completes the proof. 
\end{proof}

The following result is a generalization of Proposition \ref{bialgebra prop1} (b). 
\begin{prop}\label{compatible dual coalgebra}
	Let $(\A,[\cdot_\la \cdot]^{1},[\cdot_\la \cdot]^{2})$ be a compatible Lie conformal algebra free  of  finite rank.
	Let $\A = \bigoplus_{i=1}^n \mathbb{C}[\p]a^i$, and let  $\A^{*c} = Chom(\A,\mathbb{C}) = \bigoplus_{i=1}^n \mathbb{C}[\p]a_i$
	, where $\{a_i\}$ is a dual $\mathbb{C}[\p]$-basis in the sense that $(a_i)_\la (a^j) = \delta_{ij}$. Then  there is a compatible Lie conformal coalgebra $(\A^{*c},\de_1,\de_2)$ with the following co-brackets ($f\in \A^{*c}$):
	\begin{align*}
		\de_1(f) = \sum_{i,j}f_\mu([a^i_\la a^j]^{1})(a_i\otimes a_j)|_{\la= \p\otimes 1, \mu = -\p \otimes 1- 1\otimes \p}, ~~
		\de_2(f) = \sum_{i,j}f_\mu([a^i_\la a^j]^{2})(a_i\otimes a_j)|_{\la= \p\otimes 1, \mu = -\p \otimes 1- 1\otimes \p}.
	\end{align*}
	More precisely, if 
	\begin{align}\label{*}
		[a^i _\la a^j]^{1} = \sum_k P_k^{ij}(\la,\p)a^k, \quad
		[a^i _\la a^j]^{2} = \sum_k \tilde{P}_k^{ij}(\la,\p)a^k,
	\end{align} 
	where $P_k^{ij}$ and $\tilde{P}_k^{ij}$ are some polynomials in $\la$ and $\p$, then the co-brackets are 
	\begin{align*}
		\de_1(a_k) = \sum_{i,j}Q_k^{ij}(\p\otimes 1,1\otimes \p)a_i\otimes a_j, \quad
		\de_2(a_k) = \sum_{i,j}\tilde{Q}_k^{ij}(\p\otimes 1,1\otimes \p)a_i\otimes a_j,
	\end{align*}
	where $Q_k^{ij}(x,y)=P_k^{ij}(x,-x-y)$, $\tilde{Q}_k^{ij}(x,y)=\tilde{P}_k^{ij}(x,-x-y)$.
\end{prop}
\begin{proof}
		By Proposition \ref{bialgebra prop1}, $(\A^{*c},\de_1)$ and $(\A^{*c},\de_2)$ are Lie conformal coalgebras. It remains to show \eqref{Compatible Co} holds.
	We have
	\begin{align*}
		\de_1(a_{k}) &= \sum_{i,j} (a_{k})_{\mu} ([a^{i}_{\lambda} a^{j}]^1)(a_{i} \otimes a_{j})|_{\lambda = \partial \otimes 1, \; \mu = -\partial \otimes 1 - 1 \otimes \partial} \\
		&= \sum_{i,j} (a_{k})_{\mu} \left( \sum_{l} P^{ij}_{l} (\lambda, \partial) a^{l} \right) (a_{i} \otimes a_{j})|_{\lambda = \partial \otimes 1, \; \mu = -\partial \otimes 1 - 1 \otimes \partial} \\
		&= \sum_{i,j} P^{ij}_{k} (\lambda, \mu) (a_{i} \otimes a_{j})|_{\lambda = \partial \otimes 1, \; \mu = -\partial \otimes 1 - 1 \otimes \partial} \\
		&= \sum_{i,j} Q^{ij}_{k} (\partial \otimes 1, 1 \otimes \partial) a_{i} \otimes a_{j}.
	\end{align*}
	Similarly, we have $\de_2(a_k) = \sum_{i,j} \tilde{Q}_k^{ij}(\p\otimes 1,1\otimes \p)a_i\otimes a_j$.
		Then
		\begin{align}
			(\id\otimes \de_1)\de_2(a_k)-\tau_{12}(\id\otimes \de_1)\de_2(a_k)
			&=\sum_{i,j,l,r}\tilde{Q}_k^{ij}(\p\otimes 1\otimes 1,1\otimes \p\otimes 1+1\otimes 1\otimes \p)\nonumber\\
			&\quad\times Q_j^{lr}(1\otimes \p\otimes 1,1\otimes 1\otimes \p)(a_i\otimes a_l\otimes a_r)
			-\tau_{12}~\text{(the same term)},\label{*1}\\
			(\de_1\otimes \id)\de_2(a_k)
			&=\sum_{i,j,l,r}\tilde{Q}_k^{ij}(\p\otimes 1\otimes 1+1\otimes \p\otimes 1,1\otimes 1\otimes \p)\nonumber\\
			&\quad\times Q_i^{lr}(\p\otimes 1\otimes 1,1\otimes \p\otimes 1)(a_l\otimes a_r\otimes a_j),\label{*2}\\
			(\id\otimes \de_2)\de_1(a_k)-\tau_{12}(\id\otimes \de_2)\de_1(a_k)
			&=\sum_{i,j,l,r}Q_k^{ij}(\p\otimes 1\otimes 1,1\otimes \p\otimes 1+1\otimes 1\otimes \p)\nonumber\\
			&\quad\times \tilde{Q}_j^{lr}(1\otimes \p\otimes 1,1\otimes 1\otimes \p)(a_i\otimes a_l\otimes a_r)
			-\tau_{12}~\text{(the same term)},\label{*3}\\
			(\de_2\otimes \id)\de_1(a_k)
			&=\sum_{i,j,l,r}{Q}_k^{ij}(\p\otimes 1\otimes 1+1\otimes \p\otimes 1,1\otimes 1\otimes \p)\nonumber\\
			&\quad\times \tilde{Q}_i^{lr}(\p\otimes 1\otimes 1,1\otimes \p\otimes 1)(a_l\otimes a_r\otimes a_j).\label{*4}
		\end{align}
	Since $(\A,[\cdot_\la \cdot]^1,[\cdot_\la\cdot]^2)$ is  compatible, the compatible Jacobi identity holds for all $a^i,a^l,a^r \in \A$:
	\begin{align*}
		[a^i_\lambda [a^l_\mu a^r]^{1}]^{2}
		+ [a^i_\lambda [a^l_\mu a^r]^{2}]^{1}
		- [[a^i_\lambda a^l]^{1}_{\lambda+\mu} a^r]^{2}
		- [a^l_\mu [a^i_\lambda a^r]^{1}]^{2}
		- [[a^i_\lambda a^l]^{2}_{\lambda+\mu} a^r]^{1}
		- [a^l_\mu [a^i_\lambda a^r]^{2}]^{1} = 0.
	\end{align*}
 Expanding this by \eqref{*} gives   
	\begin{align*}
		&\sum_{j}P_j^{lr}(\mu,\la +\p)\tilde{P}_k^{ij}(\la,\p)
		+ \sum_{j}\tilde{P}_j^{lr}(\mu,\la +\p)P_k^{ij}(\la,\p)\\
		 = &\sum_j P_j^{il}(\la, -\la - \mu)\tilde{P}_k^{jr}(\la + \mu,\p)
		 + \sum_j P_j^{ir}(\la, \mu+\p)\tilde{P}_k^{lj}(\mu,\p)\\
		&+ \sum_j \tilde{P}_j^{il}(\la, -\la - \mu)P_k^{jr}(\la + \mu,\p)
		+ \sum_j \tilde{P}_j^{ir}(\la, \mu+\p)P_k^{lj}(\mu,\p).
	\end{align*}
	Since $Q_k^{ij}(x,y) = P_k^{ij}(x,-x-y)$ and $\tilde{Q}_k^{ij}(x,y) = \tilde{P}_k^{ij}(x,-x-y)$, we have 
	\begin{align*}
		&\sum_j Q_j^{lr}(\mu,-\la-\mu-\p)\tilde{Q}_k^{ij}(\la,-\la-\p)
		+ \sum_j \tilde{Q}_j^{lr}(\mu,-\la-\mu-\p)Q_k^{ij}(\la,-\la-\p) \\
		= &\sum_j Q_j^{il}(\la,\mu)\tilde{Q}_k^{jr}(\la + \mu,-\la-\mu-\p)
		+ \sum_j Q_j^{ir}(\la,-\la-\mu-\p)\tilde{Q}_k^{lj}(\mu,-\mu-\p)\\
		&+ \sum_j \tilde{Q}_j^{il}(\la,\mu)Q_k^{jr}(\la + \mu,-\la-\mu-\p)
		+ \sum_j \tilde{Q}_j^{ir}(\la,-\la-\mu-\p)Q_k^{lj}(\mu,-\mu-\p).
	\end{align*}
	Setting
		$\la = \p \otimes 1 \otimes 1,$ $ \mu = 1 \otimes \p \otimes 1, $ and $ \p = -\p \otimes 1 \otimes 1 - 1 \otimes \p \otimes 1 - 1 \otimes 1 \otimes \p$ in the above equation gives   
	\begin{align*}
		\sum_j&\tilde{Q}_k^{ij}(\p \otimes 1 \otimes 1, 1 \otimes \p \otimes 1 + 1 \otimes 1 \otimes \p)Q_j^{lr}(1 \otimes \p \otimes 1 , 1 \otimes 1 \otimes \p)\\
		& +\sum_jQ_k^{ij}(\p \otimes 1 \otimes 1 , 1 \otimes \p \otimes 1 + 1 \otimes 1 \otimes \p)\tilde{Q}_j^{lr}(1 \otimes \p \otimes 1 , 1 \otimes 1 \otimes \p)\\
		&= \sum_j\tilde{Q}_k^{jr}(\p \otimes 1 \otimes 1 + 1 \otimes \p \otimes 1, 1 \otimes 1 \otimes \p)Q_j^{il}(\p \otimes 1 \otimes 1, 1 \otimes \p \otimes 1)\\
		& + \sum_j\tilde{Q}_k^{lj}(1 \otimes \p \otimes 1, \p \otimes 1 \otimes 1 + 1 \otimes 1 \otimes \p)Q_j^{ir}(\p \otimes 1 \otimes 1, 1 \otimes 1 \otimes \p)\\
		&  + \sum_jQ_k^{jr}(\p \otimes 1 \otimes 1 + 1 \otimes \p \otimes 1, 1 \otimes 1 \otimes \p)\tilde{Q}_j^{il}(\p \otimes 1 \otimes 1, 1 \otimes \p \otimes 1)\\
		&  + \sum_jQ_k^{lj}(1 \otimes \p \otimes 1, \p \otimes 1 \otimes 1 + 1 \otimes 1 \otimes \p)\tilde{Q}_j^{ir}(\p \otimes 1 \otimes 1, 1 \otimes 1 \otimes \p).
	\end{align*}
Combining this with Eq.s \eqref{*1}--\eqref{*4}, we obtain 
	\begin{align*}
			(\id \otimes \de_1)\de_2(a_k) - \tau_{12}(\id \otimes \de_1)\de_2(a_k) - (\de_1 \otimes \id)\de_2(a_k) + (\id \otimes \de_2)\de_1(a_k) - \tau_{12}(\id \otimes \de_2)\de_1(a_k) - (\de_2 \otimes \id)\de_1(a_k) = 0.
	\end{align*}
    Thus \eqref{Compatible Co} holds, finishing the proof.  
\end{proof}
\begin{remark}\label{2dual}
    Let $\A^{*c}$ be a (compatible) Lie conformal algebra free of finite rank. There is a (compatible) Lie conformal coalgebra structure on $\A$ since $\A^{**c}\simeq \A$ as $\C[\p]$-modules.
\end{remark}

\begin{ex}\label{ex-W-dual-coalgebra}
	Let $\A=\C[\p]L\oplus \C[\p]H$. For $s=1,2$, define a $\lambda$-bracket
	$[\cdot_\lambda\cdot]^s$ on $\A$ by
	\[
		[L_\lambda L]^s=(\partial+2\lambda)L,\quad
		[L_\lambda H]^s=(\partial+(1-b_s)\lambda)H,\quad
		[H_\lambda L]^s=(-b_s\partial+(1-b_s)\lambda)H,\quad
		[H_\lambda H]^s=0,
	\]
	By Example \ref{exW1}, $(\A,[\cdot_\lambda\cdot]^s)$ is the Lie conformal algebra $W(b_s)$,
	and the two $\la$-brackets $[\cdot_\lambda\cdot]^1$ and $[\cdot_\lambda\cdot]^2$ are compatible. Let
	\[
		\A^{*c}=\C[\p]L^*\oplus \C[\p]H^*
	\]
	be the conformal dual, where
	\[
		(L^*)_\la(L)=1,\quad (L^*)_\la(H)=0,\quad
		(H^*)_\la(L)=0,\quad (H^*)_\la(H)=1.
	\]

	By Proposition \ref{compatible dual coalgebra}, they induce a compatible Lie conformal coalgebra
	$(\A^{*c},\de_1,\de_2)$, where
	\begin{align}
		\de_1(L^*)&=(\p\otimes 1-1\otimes \p)L^*\otimes L^*,\label{dual-W-coalgebra-1}\\
		\de_1(H^*)&=(-b_1\p\otimes 1-1\otimes \p)L^*\otimes H^*
		+(\p\otimes 1+b_1\,1\otimes \p)H^*\otimes L^*,\label{dual-W-coalgebra-2}\\
		\de_2(L^*)&=(\p\otimes 1-1\otimes \p)L^*\otimes L^*,\label{dual-W-coalgebra-3}\\
		\de_2(H^*)&=(-b_2\p\otimes 1-1\otimes \p)L^*\otimes H^*
		+(\p\otimes 1+b_2\,1\otimes \p)H^*\otimes L^*.\label{dual-W-coalgebra-4}
	\end{align}
	\end{ex}


\section{Compatible  Lie conformal bialgebras and Manin triples}

In this section, we introduce  compatible conformal Manin triples and compatible  Lie conformal bialgebras. The main result proves the equivalence among compatible Lie conformal bialgebras, compatible conformal Manin triples, and matched pairs of compatible Lie conformal algebras for $\C[\partial]$-modules that are free of finite rank.

\begin{defi}(\cite{BKL})\label{CBF}{\rm 
	Let $V$ be a $\C[\p]$-module. A conformal bilinear form on $V$ is a $\mathbb{C}$-bilinear map $\langle\cdot,\cdot\rangle_\la$ : $V \times V \rightarrow \mathbb{C}[\la] $ such that 
	\begin{align}
		\langle\p v,w\rangle_\la = -\lambda \langle v,w\rangle_\la= -\langle v,\p w\rangle_\la,\quad \forall v,w \in V.\label{bilinear}
	\end{align} 
	The conformal bilinear form is symmetric if $\langle v,w\rangle_\la = \langle w,v\rangle_{-\la}$ for all $v,w \in V$. 
	The conformal bilinear form in a  Lie conformal algebra $R$ is called invariant if 
	\begin{align}
        \langle[a_\mu b],c\rangle_\la = \langle a,[b_{\la-\p}c]\rangle_\mu = -\langle a,[c_{-\la}b]\rangle_\mu ,\quad \forall a,b,c \in R.\label{invariant}
	\end{align}
    The conformal bilinear form is non-degenerate if $\langle v,w \rangle_\la = 0$ for all $w\in V$, implies $v=0$.
    }
\end{defi}

\begin{defi}(\cite{Li}){\rm
	A finite conformal Manin triple is a triple of finite Lie conformal algebras $(R,R_1,R_0)$, where $R$
	is equipped with a non-degenerate invariant symmetric bilinear form $\langle \cdot , \cdot\rangle _\la $
	such that\\
	1. $R_1$, $R_0$ are Lie conformal subalgebras of $R$ and $R=R_0\oplus R_1$ as $\mathbb{C}[\p]$-module.\\
	2. $R_0$ and $R_1$ are isotropic with respect to $\langle \cdot , \cdot\rangle _\la $, that is $\langle R_i , R_i\rangle _\la = 0 $ for $i=0,1$.}
\end{defi}

\begin{defi}{\rm 
	A subalgebra of a compatible Lie conformal algebra $(\A,[\cdot_\la\cdot]^1, [\cdot_\la\cdot]^2)$ is a $\C[\p]$-submodule of $\A$  which is closed for both $\la$-brackets. }
\end{defi}

In what follows, we will write $\tilde{\A}$ to denote a compatible Lie conformal algebra $(\A,[\cdot_\la\cdot]^1, [\cdot_\la\cdot]^2)$ for simplicity.

\begin{defi}{\rm
    A finite compatible conformal Manin triple is a triple of finite compatible Lie conformal algebras $(\tilde{\A},\tilde{\A_1},\tilde{\A_0})$, where $\tilde{\A}$ is equipped with a non-degenerate symmetric conformal bilinear form $\langle\cdot,\cdot\rangle_\la$ which is invariant with respect to both compatible $\lambda$-brackets, namely, for $i=1,2$,
    \begin{align*}
        \langle[a_\mu b]^i,c\rangle_\lambda
        =\langle a,[b_{\lambda-\partial}c]^i\rangle_\mu
        =-\langle a,[c_{-\lambda}b]^i\rangle_\mu,\quad a,b,c\in\A.
    \end{align*}
    Equivalently, by linearity, the form is invariant for every linear combination $k_1[\cdot_\lambda\cdot]^1+k_2[\cdot_\lambda\cdot]^2$. Moreover,\\
    1. $\tilde{\A_1}$, $\tilde{\A_0}$ are compatible Lie conformal subalgebras of $\tilde{\A}$ and $\tilde{\A}=\tilde{\A_0}\oplus \tilde{\A_1}$ as $\mathbb{C}[\p]$-module.\\
		2. $\tilde{\A_0}$ and $\tilde{\A_1}$ are isotropic with respect to $\langle \cdot , \cdot\rangle _\la $, that is $\langle \tilde{\A_i} , \tilde{\A_i}\rangle _\la = 0 $ for $i=0,1$.}
\end{defi}

\begin{defi}{\rm
	Let $\tilde{\A}=(\A,[\cdot_\lambda\cdot]^1,[\cdot_\lambda\cdot]^2)$ be a compatible Lie conformal algebra whose underlying $\C[\partial]$-module is free of finite rank, and suppose that its conformal dual $\A^{*c}$ carries a compatible Lie conformal algebra structure $\tilde{\A}^{*c}$. If there is a compatible Lie conformal algebra structure on the direct sum of
$\tilde{\A}$ and $\tilde{\A}^{*c}$ such that $(\tilde{\A} \oplus \tilde{\A}^{*c}, \tilde{\A}, \tilde{\A}^{*c})$ is a compatible conformal Manin triple with the non-degenerate symmetric conformal bilinear form on $\tilde{\A} \oplus \tilde{\A}^{*c}$ defined by
\begin{align}
		\langle a+f,b+g \rangle _\la = f_\la(b) + g_{-\la}(a) \quad \forall~ a,b\in \A,\ f,g\in \A^{*c},\label{standard Manin}
\end{align}
then $(\tilde{\A} \oplus \tilde{\A}^{*c}, \tilde{\A}, \tilde{\A}^{*c})$ is called a standard compatible conformal Manin triple.}
\end{defi}

Let $\tilde{\A}=(\A,[\cdot_\la \cdot]^{1},[\cdot_\la \cdot]^{2})$ be a compatible Lie conformal algebra whose underlying $\C[\partial]$-module is free of finite rank, and suppose that its conformal dual carries a compatible Lie conformal algebra structure $\tilde{\A}^{*c}=(\A^{*c},\{\cdot_\la \cdot\}^{1},\{\cdot_\la \cdot\}^{2})$. From now until the end of this section, we will use $\ad$ (resp. $\mathfrak{ad}$) to denote the adjoint representations of $\A$ (resp. $\A^{*c}$), and use $\ad^*$ (resp. $\mathfrak{ad}^*$) to denote the coadjoint actions of $\A$ on $\A^{*c}$ (resp. $\A^{*c}$ on $\A$).

\begin{theo}\label{MP-triple}
	Let $\tilde{\A}=(\A,[\cdot_\la \cdot]^{1},[\cdot_\la \cdot]^{2})$ be a compatible Lie conformal algebra whose underlying $\C[\partial]$-module is free of finite rank, and suppose that its conformal dual carries a compatible Lie conformal algebra structure $\tilde{\A}^{*c}=(\A^{*c},\{\cdot_\la \cdot\}^{1},\{\cdot_\la \cdot\}^{2})$.
	Then $(\tilde{\A} \oplus \tilde{\A}^{*c}, \tilde{\A}, \tilde{\A}^{*c})$ forms a standard compatible conformal Manin triple if and only if the 6-tuple 
	$(\A,\A^{*c};\ad^*_{1},\ad^*_{2},\mathfrak{ad}^*_{1},\mathfrak{ad}^*_{2})$ is a matched pair of compatible Lie conformal algebras, where 
	\[
	\ad_i(x)_\la y = [x_\la y]^i,  x,y\in\A;~~ \mathfrak{ad}_i(f)_\la g = \{f_\la g\}^i, f,g\in\A^{*c},~~ i = 1,2.
	\] 
\end{theo}

\begin{proof}
	If $(\A,\A^{*c};\ad^*_{1},\ad^*_{2},\mathfrak{ad}^*_{1},\mathfrak{ad}^*_{2})$ is a matched pair of compatible Lie conformal algebras, then 
	there exists a compatible Lie conformal algebra $\tilde{\A} \bowtie_{\mathfrak{ad}^*_{1},\mathfrak{ad}^*_{2}}^{\ad^*_{1},\ad^*_{2}} \tilde{\A}^{*c}$  given by (cf. Theorem \ref{MP compatible})
	\begin{align}
		\relax [\!\![(a + f)_\lambda (b + g)]\!\!]^1 = [a_\lambda b]^{1} 
		+ \mathfrak{ad}^*_{1}(f)_\la b 
		- \mathfrak{ad}^*_{1}(g)_{-\lambda - \partial}a 
		 + \{f_\lambda g\}^{1} 
		+ \ad^*_{1}(a)_\lambda g 
		- \ad^*_{1}(b)_{-\lambda - \partial}f,\\
		[\!\![(a + f)_\lambda (b + g)]\!\!]^2 = [a_\lambda b]^{2} 
		+ \mathfrak{ad}^*_{2}(f)_\la b 
		- \mathfrak{ad}^*_{2}(g)_{-\lambda - \partial}a 
		 + \{f_\lambda g\}^{2} 
		+ \ad^*_{2}(a)_\lambda g 
		- \ad^*_{2}(b)_{-\lambda - \partial}f,
	\end{align}
	where $f,g\in \tilde{\A}^{*c}$ and $a,b\in \tilde{\A}$.
	The bilinear form given by Eq.\eqref{standard Manin} is symmetric because
	\[
	\langle a+f,b+g\rangle_\lambda
	=f_\lambda(b)+g_{-\lambda}(a)
	=\langle b+g,a+f\rangle_{-\lambda}.
	\]
	It is non-degenerate since, if $\langle a+f,b+g\rangle_\lambda=0$ for all $b+g$, then taking
	$g=0$ gives $f_\lambda(b)=0$ for all $b\in\A$, hence $f=0$, and taking $b=0$ gives
		$g_{-\lambda}(a)=0$ for all $g\in\A^{*c}$, hence $a=0$ by conformal duality in the finite free case. It remains to prove \eqref{invariant}.
	For all $a,b,c\in \A$ and $f,g,h\in \A^{*c}$, we have 
	\begin{align*}
		&\langle [\!\![(a+f)_\la(b+g)]\!\!]^1, c+h \rangle _\mu\\
		&= \langle [a_\lambda b]^{1} 
		+ \mathfrak{ad}^*_{1}(f)_\la b 
		- \mathfrak{ad}^*_{1}(g)_{-\lambda - \partial}a 
		 + \{f_\lambda g\}^{1} 
		+ \ad^*_{1}(a)_\lambda g 
		- \ad^*_{1}(b)_{-\lambda - \partial}f,c+h \rangle _\mu\\
		&= \{f_\lambda g\}^{1}_\mu c + (\ad^*_{1}(a)_\lambda g)_\mu c - (\ad^*_{1}(b)_{-\lambda - \partial}f)_\mu c	
		+ h_{-\mu}[a_\lambda b]^{1} + h_{-\mu}(\mathfrak{ad}^*_{1}(f)_\la b)
        - h_{-\mu}(\mathfrak{ad}^*_{1}(g)_{-\lambda - \partial}a)\\
		&= \{f_\lambda g\}^{1}_\mu c + g_{\mu- \la}[c_{-\la-\p}a]^{1} + f_\la[b_{\mu-\la}c]^{1} 
		+ h_{-\mu}[a_\la b]^{1} - \{f_\la h\}^{1}_{-\mu+\la}b + \{g_{\mu-\la}h\}^{1}_{-\la}a\\
		&= -\{g_{-\la-\p}f\}^{1}_\mu c + g_{\mu-\la}[c_{-\mu}a]^{1} + f_\la[b_{\mu-\la}c]^{1} 
		- h_{-\mu}[b_{-\la-\p}a]^{1} + \{h_{-\la-\p}f\}^{1}_{-\mu+\la}b 
        + \{g_{\mu-\la}h\}^{1}_{-\la}a\\
		&= -\{g_{\mu-\la}f\}^{1}_\mu c + g_{\mu-\la}[c_{-\mu}a]^{1} + f_\la[b_{\mu-\la}c]^{1} 
		- h_{-\mu}[b_{\mu-\la}a]^{1} + \{h_{-\mu}f\}^{1}_{-\mu+\la}b + \{g_{\mu-\la}h\}^{1}_{-\la}a,
	\end{align*}
	and
	\begin{align*}
		& \langle a+f,[\!\![(b+g)_{\mu-\p}(c+h)]\!\!]^1\rangle _\la\\
		&= \langle a+f, [b_{\mu-\p} c]^{1} + \mathfrak{ad}^*_{1}(g)_{\mu-\p} c - \mathfrak{ad}^*_{1}(h)_{-\mu}b
		+ \{g_{\mu-\p} h\}^{1} + \ad^*_{1}(b)_{\mu-\p} h - \ad^*_{1}(c)_{-\mu}g\rangle _\la \\
		&= f_\la[b_{\mu-\p} c]^{1} + f_\la(\mathfrak{ad}^*_{1}(g)_{\mu-\p} c) - f_\la(\mathfrak{ad}^*_{1}(h)_{-\mu}b)
		+ \{g_{\mu-\p} h\}^{1}_\la a + (\ad^*_{1}(b)_{\mu-\p} h)_\la a 
        - (\ad^*_{1}(c)_{-\mu}g)_\la a \\
		&= -\{g_{\mu-\la}f\}^{1}_\mu c + g_{\mu-\la}[c_{-\mu}a]^{1} + f_\la[b_{\mu-\la}c]^{1} 
		- h_{-\mu}[b_{\mu-\la}a]^{1} + \{h_{-\mu}f\}^{1}_{-\mu+\la}b + \{g_{\mu-\la}h\}^{1}_{-\la}a.
	\end{align*}
Hence $\langle [\!\![(a+f)_\la(b+g)]\!\!]^1, c+h \rangle _\mu = \langle a+f,[\!\![(b+g)_{\mu-\p}(c+h)]\!\!]^1\rangle _\la$. Similarly, we have   $\langle [\!\![(a+f)_\la(b+g)]\!\!]^2, c+h \rangle _\mu = \langle a+f,[\!\![(b+g)_{\mu-\p}(c+h)]\!\!]^2\rangle _\la$.  
    
	Now we prove the converse part. Suppose that
	$(\tilde{\A}\oplus\tilde{\A}^{*c},\tilde{\A},\tilde{\A}^{*c})$
	is a standard compatible conformal Manin triple. Denote by
	$\pi_{\A}$ and $\pi_{\A^{*c}}$ the projections from
	$\A\oplus \A^{*c}$ onto $\A$ and $\A^{*c}$, respectively. Since
	$\A$ and $\A^{*c}$ are compatible Lie conformal subalgebras and
	$\A\oplus\A^{*c}$ is their direct sum as a $\C[\p]$-module, we define four conformal linear maps by
	\[
		\rho_i^*(a)_\lambda f
		=\pi_{\A^{*c}}\big([\!\![a_\lambda f]\!\!]^i\big),
		\qquad
		\rho_i(f)_\lambda a
		=\pi_{\A}\big([\!\![f_\lambda a]\!\!]^i\big),
		\qquad i=1,2.
	\]
	Thus, by skew-symmetry, for $a,b\in\A$ and $f,g\in\A^{*c}$,
	\begin{align*}
		[\!\![(a+f)_\lambda(b+g)]\!\!]^i
		=[a_\lambda b]^i+\{f_\lambda g\}^i
		+\rho_i(f)_\lambda b-\rho_i(g)_{-\lambda-\p}a 
		+\rho_i^*(a)_\lambda g-\rho_i^*(b)_{-\lambda-\p}f,
		\qquad i=1,2.
	\end{align*}
	Since $\A\oplus\A^{*c}$ is a compatible Lie conformal algebra, by Theorem
	\ref{MP compatible},
		$(\A,\A^{*c};\rho_1^*,\rho_2^*,\rho_1,\rho_2)$
	is a matched pair of compatible Lie conformal algebras. It remains to
	identify these four actions with the coadjoint actions.

	For $a,b\in\A$, $f\in\A^{*c}$ and $i=1,2$, using the invariance of the
	standard bilinear form and the isotropy of $\A$ and $\A^{*c}$, we obtain
	\begin{align*}
		\langle \rho_i^*(a)_\lambda f,b\rangle_\mu
		&=\langle [\!\![a_\lambda f]\!\!]^i,b\rangle_\mu       
		=\langle a,[\!\![f_{\mu-\p}b]\!\!]^i\rangle_\lambda   
		=\langle f,[b_{-\lambda-\p}a]^i\rangle_{\mu-\lambda}  
		=f_{\mu-\lambda}([b_{-\mu}a]^i)\\                       
		&=-f_{\mu-\lambda}([a_{\mu-\p}b]^i)                    
		=(\ad_i^*(a)_\lambda f)_\mu b
		=\langle \ad_i^*(a)_\lambda f,b\rangle_\mu .
	\end{align*}
	Hence, by the non-degeneracy of \eqref{standard Manin}, it follows that
		$\rho_i^*=\ad_i^*$, $i=1,2$.
	For the action of $\A^{*c}$ on $\A$, the same invariance identity gives a second explicit computation.
	For $f,g\in\A^{*c}$, $a\in\A$ and $i=1,2$,
	\begin{align*}
		\langle \rho_i(f)_\lambda a,g\rangle_\mu
		&=\langle [\!\![f_\lambda a]\!\!]^i,g\rangle_\mu        
		=\langle f,[\!\![a_{\mu-\p}g]\!\!]^i\rangle_\lambda   
		=\langle a,\{g_{-\lambda-\p}f\}^i\rangle_{\mu-\lambda}\\
		&=-\langle a,\{f_{\mu-\p}g\}^i\rangle_{\mu-\lambda}    
		=(\mathfrak{ad}_i^*(f)_\lambda a)_\mu g
		=\langle \mathfrak{ad}_i^*(f)_\lambda a,g\rangle_\mu .
	\end{align*}
	Again by non-degeneracy,
		$\rho_i=\mathfrak{ad}_i^*$, $i=1,2$.

	Therefore
		$(\A,\A^{*c};\ad_1^*,\ad_2^*,\mathfrak{ad}_1^*,\mathfrak{ad}_2^*)$
	is a matched pair of compatible Lie conformal algebras. This proves the
	converse implication.
\end{proof}

\begin{coro}
    Let $(\A,[\cdot_\la\cdot])$ be a Lie conformal algebra whose underlying $\C[\partial]$-module is free of finite rank, and suppose that its conformal dual carries a Lie conformal algebra structure $(\A^{*c},\{\cdot_\la\cdot\})$. Then $(\A\oplus\A^{*c},\A,\A^{*c})$ forms a standard conformal Manin triple if and only if $(\A,\A^{*c};\ad^*,\mathfrak{ad}^*)$ is a matched pair of LCAs.
\end{coro}
\begin{proof}
    It is a special case of Theorem \ref{MP-triple}.
\end{proof}

We combine the corollary above with \cite[Theorem 3.9]{Li} into a single theorem, and then proceed to prove it in the case of compatible Lie conformal algebras.

\begin{theo}\label{conformal Matched pair and Manin and bialgebra}(\cite{Li})
		Let $(\A,[\cdot_\la \cdot])$ be a Lie conformal algebra whose underlying $\C[\partial]$-module is free of finite rank. Suppose that its conformal dual carries a Lie conformal algebra structure $(\A^{*c},\{\cdot_\la\cdot\})$. Denote by $\de$ the Lie conformal coalgebra structure induced on $\A$ by Proposition \ref{bialgebra prop1} (cf. Remark \ref{2dual}).
	Then the following conditions are equivalent:
	\begin{itemize}[(c)]
	\item[(a)] $(\A,\A^{*c};\ad^*,\mathfrak{ad}^*)$ is a matched pair of Lie conformal algebras, where $\ad$ and $\mathfrak{ad}$ are the adjoint representations of $\A$ and $\A^{*c}$, respectively;
	\item[(b)] $(\A \oplus \A^{*c},\A,\A^{*c})$ is a standard Manin triple of Lie conformal algebras $\A$ and $\A^{*c}$;
	\item[(c)] $(\A,[\cdot _\la \cdot], \delta)$ is a Lie conformal bialgebra.
	\end{itemize}
\end{theo}


Now we define compatible Lie conformal bialgebras.

\begin{defi}\label{com bia}{\rm Two Lie conformal bialgebras  $(\A,[\cdot_\la\cdot]^1,\de_1)$ and $(\A,[\cdot_\la\cdot]^2,\de_2)$ 
			are called compatible, if
	\begin{itemize}[(3)]
		\item[(i)] 		
	 			$(\A,[\cdot_\la\cdot]^1,[\cdot_\la\cdot]^2)$ is a compatible Lie conformal algebra,
		\item[(ii)] $(\A,\de_1,\de_2)$ is a compatible Lie conformal coalgebra,
		\item[(iii)] $(\A,[\cdot_\la\cdot]^1+[\cdot_\la\cdot]^2,\de_1+\de_2)$ is a Lie conformal bialgebra.		
		\end{itemize} }
\end{defi}

If $(\A,[\cdot_\la\cdot]^1,\de_1)$ and $(\A,[\cdot_\la\cdot]^2,\de_2)$  are compatible Lie conformal bialgebras, then we denote it by the five-tuple $(\A,[\cdot_\la \cdot]^1,[\cdot_\la \cdot]^2,\de_1,\de_2)$.
\begin{remark}
    Let $(\A,[\cdot_\la\cdot]^1,\de_1)$ and $(\A,[\cdot_\la\cdot]^2,\de_2)$ be two Lie conformal bialgebras. The following statements are equivalent:
    \begin{itemize}[(iii)]
        \item[(i)] $(\A,[\cdot_\la \cdot]^1,[\cdot_\la \cdot]^2,\de_1,\de_2)$ is a compatible Lie conformal bialgebra.
        \item[(ii)] The compatible Jacobi identity \eqref{compatible Jacobi}, the compatible co-Jacobi identity \eqref{Compatible Co}, and the following compatible cocycle identity hold for all $x,y\in \A$:
        \begin{align}\label{compatible cocycle}
            &\de_1 ([x_\la y]^{2}) + \de_2 ([x_\la y]^{1}) \nonumber\\
            &=(\ad_{2}(x)_\la \otimes \id + \id \otimes \ad_{2}(x)_\la)\de_1(y)
            -(\ad_{2}(y)_{-\la -\p}\otimes \id + \id \otimes \ad_{2}(y)_{-\la -\p})\de_1(x) \nonumber\\
            &\quad+(\ad_{1}(x)_\la \otimes \id + \id \otimes \ad_{1}(x)_\la)\de_2(y)
            -(\ad_{1}(y)_{-\la -\p}\otimes \id + \id \otimes \ad_{1}(y)_{-\la -\p})\de_2(x).
        \end{align}
        \item[(iii)] For all $k_1,k_2\in\C$,
            $\big(\A,k_1[\cdot_\lambda\cdot]^1+k_2[\cdot_\lambda\cdot]^2,
            k_1\de_1+k_2\de_2\big)$
        is a Lie conformal bialgebra.
    \end{itemize}
\end{remark}

\begin{ex}\label{ex-Rp-compatible-bialgebra}
	Keep the notation in Examples \ref{ex-Rp-compatible-LCA} and \ref{ex-Rp-compatible-coalgebra}. We show that
	$(R,[\cdot_\lambda\cdot]^1,[\cdot_\lambda\cdot]^2,\de_1,\de_2)$
	is a compatible Lie conformal bialgebra.

	By \cite[Example 2.18]{Li}, for the rank two solvable Lie conformal algebra $R_p$ with
	$[a_\lambda b]=p(\lambda)b$ and the cobracket
	$\de_q(a)=0$, $\de_q(b)=q(\p)(a\wedge b)$, the triple
	$(R_p,[\cdot_\lambda\cdot],\de_q)$ is a Lie conformal bialgebra if and only if
	$p(x)q(x)$ is an odd polynomial. In the first structure we have
	$$
		p_1(x)=1,\qquad q_1(x)=x,
	$$
	and hence $p_1(x)q_1(x)=x$. In the second structure we have
	$$
		p_2(x)=x^2,\qquad q_2(x)=x^3,
	$$
	and hence $p_2(x)q_2(x)=x^5$. Therefore
	$(R,[\cdot_\lambda\cdot]^1,\de_1)$ and $(R,[\cdot_\lambda\cdot]^2,\de_2)$ are Lie conformal bialgebras.

	We now check the three conditions in Definition \ref{com bia}. First,
	$(R,[\cdot_\lambda\cdot]^1,[\cdot_\lambda\cdot]^2)$ is a compatible Lie conformal algebra by Example \ref{ex-Rp-compatible-LCA}. Second,
	$(R,\de_1,\de_2)$ is a compatible Lie conformal coalgebra by Example \ref{ex-Rp-compatible-coalgebra}. Finally, the sum bracket and the sum cobracket correspond to
	$$
		p(x)=1+x^2,\qquad q(x)=x+x^3.
	$$
	The product is
	\[
		p(x)q(x)=(1+x^2)(x+x^3)=x+2x^3+x^5,
    \]
		which is an odd polynomial. Thus the sum bracket together with the sum cobracket gives a
		Lie conformal bialgebra. Hence
		$(R,[\cdot_\lambda\cdot]^1,[\cdot_\lambda\cdot]^2,\de_1,\de_2)$ is a compatible Lie conformal bialgebra.
\end{ex}

Example \ref{ex-Rp-compatible-bialgebra} gives a compatible Lie conformal bialgebra obtained from the general bialgebra criterion, rather than from the coboundary construction considered in Section 6.

\begin{prop}\label{prop:structure-constant-cocycle}
	Let $\tilde{\A} = (\A,[\cdot_\la \cdot]^{1},[\cdot_\la \cdot]^{2})$ be a compatible Lie conformal algebra whose underlying $\C[\partial]$-module is free of finite rank.
	Suppose that its conformal dual carries a compatible Lie conformal algebra structure $\tilde{\A}^{*c} = (\A^{*c},\{\cdot_\la\cdot\}^1,\{\cdot_\la\cdot\}^2)$.
	Let $\de_1$ and $\de_2$ be the compatible Lie conformal coalgebra structure induced on $\A\simeq \A^{**c}$ by Proposition \ref{compatible dual coalgebra}.
	Then the compatible Jacobi identity for
	$\tilde{\A}\bowtie_{\mathfrak{ad}^*_1,\mathfrak{ad}^*_2}^{\ad^*_1,\ad^*_2}\tilde{\A}^{*c}$
	(cf. Theorem \ref{MP compatible})
	is equivalent to the compatible cocycle identity \eqref{compatible cocycle}.
\end{prop}
\begin{proof}
	Let $\{e_i\}_{i=1}^n$ be a $\C[\p]$-basis of $\A$ and $\{e_i^*\}_{i=1}^n$ be a $\C[\p]$-basis of $\A^{*c}$. We suppose  
	\begin{align*}
		[{e_i}_{\la} e_j]^{1} = \sum_s A_{ij}^s(\la,\p)e_s, \quad \{{e_i^*}_{\la} e_j^*\}^{1} = \sum_s B_s^{ij}(\la,\p)e_s^{*},\\
		[{e_i}_{\la} e_j]^{2} = \sum_s \tilde{A}_{ij}^s(\la,\p)e_s, \quad \{{e_i^*} _\la e_j^*\}^{2} = \sum_s \tilde{B}_s^{ij}(\la,\p)e_s^*.
	\end{align*}
	Since 
	\begin{align*}
			(\ad_{1}^*(e_j)_{\la} {e_i^*})_{\mu} e_k = -{e_i^*}_{\mu-\la}(\ad_{1}(e_j)_\la e_k) = -{e_i^*}_{\mu-\la}\Big(\sum_s A_{jk}^s(\la, \p)e_s\Big) = -A_{jk}^i(\la, \mu-\la),
	\end{align*}
	 we have $\ad_{1}^*(e_j)_\la e_i^* = -\sum_{k=1}^n A_{jk}^i(\la,-\la-\p)e_k^*$. Similarly, 
	\begin{align*}
	\mathfrak{ad}_{1}^*(e_i^*)_\la(e_j) &= -\sum_{k=1}^n B_j^{ik}(\la,-\la - \p)e_k = -\sum_{k=1}^n C_j^{ik}(\la,\p)e_k;\\
	\ad_{2}^*(e_j)_\la e_i^* &= -\sum_{k=1}^n \tilde{A}_{jk}^i(\la,-\la-\p)e_k^*;\\
		\mathfrak{ad}_{2}^*(e_i^*)_\la(e_j) &= -\sum_{k=1}^n \tilde{B}_j^{ik}(\la,-\la - \p)e_k = -\sum_{k=1}^n \tilde{C}_j^{ik}(\la,\p)e_k,
	\end{align*}
		where $C(\la,\p) = B(\la, -\la-\p)$ and $\tilde{C}(\la,\p) = \tilde{B}(\la, -\la-\p)$.

		For $\tilde{\A}\bowtie_{\mathfrak{ad}^*_1,\mathfrak{ad}^*_2}^{\ad^*_1,\ad^*_2}\tilde{\A}^{*c}$, the compatible
		$\la$-brackets are
		\begin{align*}
			[\!\![{e_i^*}_\la e_j]\!\!]^1 = \mathfrak{ad}_{1}^*(e_i^*)_\la e_j - \ad^*_{1}(e_j)_{-\la -\p}e_i^*,\quad
			[\!\![{e_i^*} _\la e_j]\!\!]^2 = \mathfrak{ad}_{2}^*(e_i^*)_\la e_j - \ad^*_{2}(e_j)_{-\la -\p}e_i^*.
		\end{align*}
		Therefore, we get 
		\begin{align*}
			[\!\![{e_i^*} _\la e_j]\!\!]^1 = \sum_{k=1}^n \left(A_{jk}^i(-\la -\p,\la)e_k^* - C_j^{ik}(\la,\p)e_k \right),\quad
			[\!\![{e_i^*} _\la e_j]\!\!]^2 = \sum_{k=1}^n \left(\tilde{A}_{jk}^i(-\la -\p,\la)e_k^* - \tilde{C}_j^{ik}(\la,\p)e_k \right).
		\end{align*}
		The corresponding compatible Jacobi identity for $e_p^*,e_k,e_l$ reads
		\begin{align*}
			[\!\![{e_p^*}_{\lambda}[\!\![e_k {}_{\mu} e_l]\!\!]^1]\!\!]^2
			-[\!\![[\!\![{e_p^*} {}_{\lambda} e_k]\!\!]^1 {}_{\lambda + \mu} e_l]\!\!]^2
			-[\!\![e_k {}_{\mu} [\!\![{e_p^*} {}_{\lambda} e_l]\!\!]^1]\!\!]^2
			+[\!\![{e_p^*}_{\lambda}[\!\![e_k {}_{\mu} e_l]\!\!]^2]\!\!]^1
			-[\!\![[\!\![{e_p^*} {}_{\lambda} e_k]\!\!]^2 {}_{\lambda + \mu} e_l]\!\!]^1
			-[\!\![e_k {}_{\mu} [\!\![{e_p^*} {}_{\lambda} e_l]\!\!]^2]\!\!]^1 = 0.
		\end{align*}
	Expanding it, we get 
	\begin{align}\label{jacobi and cocycle}
		0 = &\sum_{t,i}A_{kl}^i(\mu,\la +\p)(\tilde{A}_{it}^p(-\la -\p,\la)e_t^* - \tilde{C}_i^{pt}(\la,\p)e_t) - \sum_{t,i}A_{ki}^p(\mu,\la)(\tilde{A}_{lt}^i(-\la -\mu -\p,\la + \mu)e_t^* - \tilde{C}_l^{it}(\la + \mu,\p)e_t)\nonumber \\
		&+\sum_{t,i}\tilde{A}_{il}^t(\la +\mu,\p)C_k^{pi}(\la,-\la-\mu)e_t+ \sum_{t,i}A_{li}^p(-\la-\mu-\p,\la)(\tilde{A}_{kt}^i(\mu, -\mu-\p)e_t^* - \tilde{C}_k^{it}(-\mu -\p,\p)e_t)   \nonumber\\
		&+ \sum_{t,i}\tilde{A}_{ki}^t(\mu,\p)C_l^{pi}(\la,\mu+\p)e_t
		+ \sum_{t,i}\tilde{A}_{kl}^i(\mu,\la +\p)(A_{it}^p(-\la -\p,\la)e_t^* - C_i^{pt}(\la,\p)e_t) \nonumber\\
		&+ \sum_{t,i}A_{ki}^t(\mu,\p)\tilde{C}_l^{pi}(\la,\mu+\p)e_t \nonumber\\
		&- \sum_{t,i}\tilde{A}_{ki}^p(\mu,\la)(A_{lt}^i(-\la -\mu -\p,\la + \mu)e_t^* - C_l^{it}(\la + \mu,\p)e_t) + \sum_{t,i}A_{il}^t(\la +\mu,\p)\tilde{C}_k^{pi}(\la,-\la-\mu)e_t \nonumber\\
		& + \sum_{t,i}\tilde{A}_{li}^p(-\la-\mu-\p,\la)(A_{kt}^i(\mu, -\mu-\p)e_t^* - C_k^{it}(-\mu -\p,\p)e_t).
	\end{align}
			The coefficients of $e_t^*$ in \eqref{jacobi and cocycle} give the already satisfied compatible Jacobi
			identity of $\tilde{\A}\ltimes_{\ad_1^*,\ad_2^*}\A^{*c}$, because they contain only the
			structure functions $A,\tilde A$ and the terms
			$A_{jk}^i(\lambda,-\lambda-\partial)$,
				$\tilde A_{jk}^i(\lambda,-\lambda-\partial)$ arising from the coadjoint actions
				$\ad_1^*$ and $\ad_2^*$ of $\A$ on $\A^{*c}$. The coefficients of $e_t$ become \eqref{coefficients} after the change of variables $\la = \p\otimes1$, $\p=\p\otimes1$ and $\mu=\la$. 
			The compatible Jacobi identity involving
			two elements of $\A^{*c}$ and one element of $\A$ is treated in the same way, with
			$(A,\tilde A)$ and $(C,\tilde C)$ interchanged.

			Now, we write \eqref{compatible cocycle} in the above structure functions. By Proposition \ref{compatible dual coalgebra}, we have
	\begin{align*}
		\de_1(e_i) = \sum_{k,l}C_i^{kl}(\p \otimes 1,1\otimes \p)e_k \otimes e_l, \quad
		\de_2(e_i) = \sum_{k,l}\tilde{C}_i^{kl}(\p \otimes 1,1\otimes \p)e_k \otimes e_l,
	\end{align*}
		where $C(x,y) = B(x,-x-y)$ and $\tilde{C}(x,y)=\tilde{B}(x,-x-y)$. Then we have
		\begin{align}
			&\de_1([{e_k}_\la {}{e_l}]^{2}) = \sum_i \tilde{A}_{kl}^i(\la, \p \otimes 1 + 1\otimes \p)\de_1(e_i) = \sum_{i,p,q}\tilde{A}_{kl}^i(\la, \p \otimes 1 + 1 \otimes \p)C_i^{pq}(\p \otimes 1, 1 \otimes \p)e_p \otimes e_q,\\
			&\de_2([{e_k} _\la {e_l}]^{1}) = \sum_i A_{kl}^i(\la, \p \otimes 1 + 1\otimes \p)\de_2(e_i)
									 = \sum_{i,p,q}A_{kl}^i(\la, \p \otimes 1 + 1 \otimes \p)\tilde{C}_i^{pq}(\p \otimes 1, 1 \otimes \p)e_p \otimes e_q,\\
			&(\ad_{2}(e_k)_\la\otimes \id + \id \otimes \ad_{2}(e_k)_\la)\de_1(e_l)\nonumber\\
			&= (\ad_{2}(e_k)_\la\otimes \id + \id \otimes \ad_{2}(e_k)_\la)(\sum_{p,q}C_l^{pq}(\p \otimes 1,1\otimes \p)e_p \otimes e_q)\nonumber\\
			 &= \sum_{p,q,i}\Big(C_l^{pq}(\la + \p \otimes 1, 1 \otimes \p)\tilde{A}_{kp}^i(\la ,\p \otimes 1)e_i \otimes e_q 
			 + C_l^{pq}(\p \otimes 1, \la + 1 \otimes \p)\tilde{A}_{kq}^i(\la, 1 \otimes \p)e_p \otimes e_i\Big),\\
			&(\ad_1(e_k)_\la\otimes \id + \id \otimes \ad_{1}(e_k)_\la)\de_2(e_l)\nonumber\\
			& = (\ad_{1}(e_k)_\la\otimes \id + \id \otimes \ad_{1}(e_k)_\la)(\sum_{p,q}\tilde{C}_l^{pq}(\p \otimes 1,1\otimes \p)e_p \otimes e_q)\nonumber\\
			& = \sum_{p,q,i}\Big(\tilde{C}_l^{pq}(\la + \p \otimes 1, 1 \otimes \p)A_{kp}^i(\la ,\p \otimes 1)e_i \otimes e_q + \tilde{C}_l^{pq}(\p \otimes 1, \la + 1 \otimes \p)A_{kq}^i(\la, 1 \otimes \p)e_p \otimes e_i\Big),\\
			&(\ad_{2}(e_l)_{-\la -\p}\otimes \id + \id \otimes \ad_{2}(e_l)_{-\la -\p})\de_1(e_k)\nonumber\\
			& = (\ad_{2}(e_l)_{-\la -\p}\otimes \id + \id \otimes \ad_{2}(e_l)_{-\la -\p})(\sum_{p,q}C_k^{pq}(\p \otimes 1, 1 \otimes \p)e_p \otimes e_q)\nonumber\\
			& = \sum_{p,q,i}\Big(C_k^{pq}(-\la -1 \otimes \p, 1 \otimes \p)\tilde{A}_{lp}^i(-\la -\p     \otimes 1 - 1 \otimes \p, \p \otimes 1)e_i \otimes e_q\nonumber\\
				 &~~~~~~~~~~~~~~~+ C_k^{pq}(\p \otimes 1, -\la - \p \otimes 1)\tilde{A}_{lq}^i(-\la -\p \otimes 1 - 1\otimes \p, 1 \otimes \p)e_p \otimes e_i\Big),\\
			&(\ad_{1}(e_l)_{-\la -\p}\otimes \id + \id \otimes \ad_{1}(e_l)_{-\la -\p})\de_2(e_k)\nonumber\\
			& = (\ad_{1}(e_l)_{-\la -\p}\otimes \id + \id \otimes \ad_{1}(e_l)_{-\la -\p})(\sum_{p,q}\tilde{C}_k^{pq}(\p \otimes 1, 1 \otimes \p)e_p \otimes e_q)\nonumber\\
			& = \sum_{p,q,i}\Big(\tilde{C}_k^{pq}(-\la -1 \otimes \p, 1 \otimes \p)A_{lp}^i(-\la -\p \otimes 1 - 1 \otimes \p, \p \otimes 1)e_i \otimes e_q\nonumber\\
			 &~~~~~~~~~~~~~~+ \tilde{C}_k^{pq}(\p \otimes 1, -\la - \p \otimes 1)A_{lq}^i(-\la -\p \otimes 1 - 1\otimes \p, 1 \otimes \p)e_p \otimes e_i\Big).
		\end{align}
	By taking the coefficients of $e_p \otimes e_q$, the cocycle condition becomes
		\begin{align}\label{coefficients}
			&\sum_i \tilde{A}_{kl}^i(\la, \p \otimes 1 + 1 \otimes \p)C_i^{pq}(\p \otimes 1, 1 \otimes \p) + \sum_{i}A_{kl}^i(\la, \p \otimes 1 + 1 \otimes \p)\tilde{C}_i^{pq}(\p \otimes 1, 1 \otimes \p)\nonumber\\
			& \quad =\sum_i \Big(C_l^{iq}(\la + \p \otimes 1, 1 \otimes \p)\tilde{A}_{ki}^p(\la ,\p \otimes 1) + C_l^{pi}(\p \otimes 1, \la + 1 \otimes \p)\tilde{A}_{ki}^q(\la, 1 \otimes \p)\Big)\nonumber\\
			& \quad \quad - \sum_i \Big(C_k^{iq}(-\la -1 \otimes \p, 1 \otimes \p)\tilde{A}_{li}^p(-\la -\p \otimes 1 - 1 \otimes \p, \p \otimes 1) \nonumber\\
			& \quad \quad + C_k^{pi}(\p \otimes 1, -\la - \p \otimes 1)\tilde{A}_{li}^q(-\la -\p \otimes 1 - 1\otimes \p, 1 \otimes \p)\Big)\nonumber\\
			& \quad \quad + \sum_i \Big(\tilde{C}_l^{iq}(\la + \p \otimes 1, 1 \otimes \p)A_{ki}^p(\la ,\p \otimes 1) + \tilde{C}_l^{pi}(\p \otimes 1, \la + 1 \otimes \p)A_{ki}^q(\la, 1 \otimes \p)\Big)\nonumber\\
			& \quad \quad - \sum_i \Big(\tilde{C}_k^{iq}(-\la -1 \otimes \p, 1 \otimes \p)A_{li}^p(-\la -\p \otimes 1 - 1 \otimes \p, \p \otimes 1) \nonumber\\
			& \quad \quad + \tilde{C}_k^{pi}(\p \otimes 1, -\la - \p \otimes 1)A_{li}^q(-\la -\p \otimes 1 - 1\otimes \p, 1 \otimes \p)\Big),
		\end{align}
			which is equivalent (after renaming the variables: $\p\otimes1=\la$, $\la=\mu$, $1\otimes\p=\p$) to the coefficients of 
			$e_t$ in \eqref{jacobi and cocycle}. This proves that the compatible Jacobi identity for
			$\tilde{\A}\bowtie_{\mathfrak{ad}^*_1,\mathfrak{ad}^*_2}^{\ad^*_1,\ad^*_2}\tilde{\A}^{*c}$
			and the compatible cocycle identity \eqref{compatible cocycle} are equivalent. 
\end{proof}

\begin{theo}\label{thm:compatible-matched-pair-bialgebra}
		Let $\tilde{\A} = (\A,[\cdot_\la \cdot]^{1},[\cdot_\la \cdot]^{2})$ be a compatible Lie conformal algebra whose underlying $\C[\partial]$-module is free of finite rank. 
    Suppose that its conformal dual carries a compatible Lie conformal algebra structure $\tilde{\A}^{*c} = (\A^{*c},\{\cdot_\la\cdot\}^1,\{\cdot_\la\cdot\}^2)$. Denote by $\de_1$ and $\de_2$ the compatible Lie conformal coalgebra structure induced on $\A\simeq \A^{**c}$ by Proposition \ref{compatible dual coalgebra} (cf. Remark \ref{2dual}). Then $(\A,\A^{*c};\ad^*_{1},\ad^*_{2},\mathfrak{ad}^*_{1},\mathfrak{ad}^*_{2})$
	is a matched pair of compatible Lie conformal algebras if and only if $(\A,[\cdot _\la \cdot]^{1},[\cdot _\la \cdot]^{2},\de_1,\de_2)$ is a compatible Lie conformal bialgebra.
\end{theo}
\begin{proof}
	By Theorem \ref{conformal Matched pair and Manin and bialgebra}, for $s=1,2$,
	$(\A,\A^{*c};\ad_s^*,\mathfrak{ad}_s^*)$ is a matched pair of LCAs if and only if
	$(\A,[\cdot_\la\cdot]^s,\de_s)$ is a Lie conformal bialgebra.
	The compatibility of the two brackets on $\A$ is assumed, and the compatibility of
	$(\A,\de_1,\de_2)$ follows from Proposition \ref{compatible dual coalgebra}.
	Thus, by Theorem \ref{MP compatible} and Definition \ref{com bia}, the only
	remaining comparison is between the compatible Jacobi identity for
	$\tilde{\A}\bowtie_{\mathfrak{ad}^*_1,\mathfrak{ad}^*_2}^{\ad^*_1,\ad^*_2}\tilde{\A}^{*c}$ and
	the compatible cocycle identity \eqref{compatible cocycle}. This is exactly Proposition
	\ref{prop:structure-constant-cocycle}. The theorem follows.
\end{proof}
\begin{coro}\label{compatible conformal Matched pair and Manin and bialgebra}
		Let $\tilde{\A} = (\A,[\cdot_\la \cdot]^{1},[\cdot_\la \cdot]^{2})$ be a compatible Lie conformal algebra whose underlying $\C[\partial]$-module is free of finite rank. 
    Suppose that its conformal dual carries a compatible Lie conformal algebra structure $\tilde{\A}^{*c} = (\A^{*c},\{\cdot_\la\cdot\}^1,\{\cdot_\la\cdot\}^2)$. Denote by $\de_1$ and $\de_2$ the compatible Lie conformal coalgebra structure induced on $\A\simeq \A^{**c}$ by Proposition \ref{compatible dual coalgebra} (cf. Remark \ref{2dual}). 
	Then the following conditions are equivalent:
    \begin{itemize}[(c)]
	\item[(a)] $(\tilde{\A},\tilde{\A}^{*c};\ad^*_{1},\ad^*_{2},\mathfrak{ad}^*_{1},\mathfrak{ad}^*_{2})$ is a matched pair of compatible Lie conformal algebras;
    \item[(b)] $(\tilde{\A} \oplus \tilde{\A}^{*c},\tilde{\A},\tilde{\A}^{*c})$ is a standard Manin triple of compatible Lie conformal algebras $\A$ and $\A^{*c}$;
	\item[(c)] $(\A,[\cdot_\la\cdot]^1,[\cdot_\la\cdot]^2,\de_1,\de_2)$ is a compatible Lie conformal bialgebra.
    \end{itemize}
\end{coro}
\begin{coro}\label{dual compatible bialgebra}
	If $(\A,[\cdot_\la \cdot]^{1},[\cdot_\la \cdot]^{2},\de_1,\de_2)$ is a compatible Lie conformal bialgebra and $\A$ is free of finite rank, then there is a compatible Lie conformal bialgebra structure on $\A^{*c}$. 
\end{coro}

\begin{proof}
	By Proposition \ref{compatible dual}, the cobrackets $\de_1$ and $\de_2$ induce compatible Lie conformal brackets on $\A^{*c}$. Since $\A$ is free of finite rank, the brackets $[\cdot_\la\cdot]^1$ and $[\cdot_\la\cdot]^2$ induce compatible Lie conformal coalgebra structures on $\A^{*c}$ by Proposition \ref{compatible dual coalgebra}. 
	By Corollary \ref{compatible conformal Matched pair and Manin and bialgebra}, the compatible bialgebra structure on $\A$ is equivalent to the matched pair
	$(\tilde{\A},\tilde{\A}^{*c};\ad^*_{1},\ad^*_{2},\mathfrak{ad}^*_{1},\mathfrak{ad}^*_{2})$, which can be regarded as the matched pair
	$(\tilde{\A}^{*c},\tilde{\A};\mathfrak{ad}^*_{1},\mathfrak{ad}^*_{2},\ad^*_{1},\ad^*_{2})$ under the identification $\A\simeq\A^{**c}$.
	Applying the Corollary \ref{compatible conformal Matched pair and Manin and bialgebra} again yields a compatible Lie conformal bialgebra structure on $\A^{*c}$.
\end{proof}

\section{Coboundary compatible Lie conformal bialgebras}

In this section, we characterize coboundary compatible Lie conformal bialgebras defined by a tensor $r$. The characterization consists of the invariant symmetric-part condition for both brackets and the three conformal Yang--Baxter conditions in Theorem \ref{com YB}. We introduce the compatible conformal CYBE for comparison. The last two examples show, respectively, that the first two conformal Yang--Baxter conditions do not imply the third and that all three conditions do not imply the compatible conformal CYBE.
\begin{defi}(\cite{Li}){\rm
	A coboundary Lie  conformal  bialgebra is a triple $(\A,[\cdot_\la \cdot],r)$, with $r\in \A\otimes \A$, such that
	$(\A,[\cdot_\la \cdot],dr)$ is a Lie  conformal bialgebra, where $dr:= (dr)_{-\p^{\otimes 2}}$ such that
	\[
	dr(a)=(dr)_\la (a)\mid _{\la = -\p^{\otimes 2}}=a_\la r\mid _{\la = -\p^{\otimes 2}},
	\]
	for all $a\in\A$ and with $\p^{\otimes 2}=\p\otimes 1+1\otimes\p.$
	In this case, the element $r\in \A \otimes \A$ is said to be a coboundary structure.}
\end{defi}

Let $r\in \A\otimes \A$, with $r = \sum_i a_i\otimes b_i$. Define
\begin{align*}
	[\![r,r]\!] = \sum_{ij}\left([a_i {}_{\mu} a_j]\otimes b_i \otimes b_j\mid_{\mu = 1\otimes \partial \otimes 1} - a_i \otimes [a_j {}_{\mu} b_i]\otimes b_j\mid_{\mu = 1\otimes 1 \otimes \partial} - a_i \otimes a_j \otimes [b_j {}_{\mu} b_i]\mid_{\mu = 1\otimes \partial \otimes 1}\right).
\end{align*}

\begin{theo}(\cite{Li})\label{thm:Li-coboundary}
	Let $\A$ be a Lie conformal algebra and let $r\in \A \otimes \A$. The map
	\[
	\delta(a)=(dr)_\la (a)=a_\la r\mid_{\la = -\p^{\otimes 2}}, \forall a\in \A,
	\]
	is the cocommutator of a Lie conformal bialgebra structure on $\A$ if and only if the following conditions are satisfied:
    \begin{itemize}[(b)]
	\item[(a)] the symmetric part of $r$ is $\A$-invariant, that is:
	\begin{align*}
		a_\la(r+r^{21})\mid_{\la = -\p^{\otimes 2}} = 0,
	\end{align*}
		where $r^{21}=\sum_i b_i \otimes a_i$, if $r = \sum_i a_i \otimes b_i$.
	\item[(b)] $a_\la[\![r,r]\!]\mid_{\la=-\p^{\otimes 3}} = 0$, where $\p^{\otimes 3} = \p \otimes 1 \otimes 1 + 1 \otimes \p \otimes 1 + 1 \otimes 1 \otimes \p$.
    \end{itemize}
\end{theo}

Now we define coboundary compatible Lie conformal bialgebras:
\begin{defi}\label{com cob}{\rm
	A coboundary compatible Lie conformal bialgebra is a quadruple $(\A,[\cdot_\la \cdot]^{1},[\cdot_\la \cdot]^{2},r)$, with $r = \sum_i a_i \otimes b_i \in \A \otimes \A$,
	such that $(\A,[\cdot_\la\cdot]^1,[\cdot_\la\cdot]^2,d_1 r,d_2 r)$ is a compatible Lie conformal bialgebra, with 
    \begin{align}  
    (d_1 r)_\la (a)\mid _{\la = -\p^{\otimes 2}} = (a_\la r)^{1}\mid_{\la = -\p^{\otimes 2}} = \sum_i ([a_\la a_i]^{1}\otimes b_i + a_i \otimes [a_\la b_i]^{1})\mid_{\la = -\p^{\otimes 2}},\label{coboundary1}\\
	(d_2 r)_\la (a)\mid _{\la = -\p^{\otimes 2}} = (a_\la r)^{2}\mid_{\la = -\p^{\otimes 2}} = \sum_i ([a_\la a_i]^{2}\otimes b_i + a_i \otimes [a_\la b_i]^{2})\mid_{\la = -\p^{\otimes 2}}. \label{coboundary2} 
    \end{align}}
   \end{defi}

Let $(\A,[\cdot_\la \cdot]^{1},[\cdot_\la \cdot]^{2})$ be a compatible LCA and 
let $r\in \A\otimes \A$, with $r = \sum_i a_i\otimes b_i$. Define
    \begin{align}
        [\![r,r]\!]^{1} &= \sum_{i,j}\Big([a_i {}_{\mu} a_j]^{1}\otimes b_i \otimes b_j\mid_{\mu = 1\otimes \partial \otimes 1}
        - a_i \otimes [a_j {}_{\mu} b_i]^{1}\otimes b_j\mid_{\mu = 1\otimes 1 \otimes \partial} - a_i \otimes a_j \otimes [b_j {}_{\mu} b_i]^{1}\mid_{\mu = 1\otimes \partial \otimes 1}\Big),\label{YB1}\\
	 [\![r,r]\!]^{2} &= \sum_{i,j}\Big([a_i {}_{\mu} a_j]^{2}\otimes b_i \otimes b_j\mid_{\mu = 1\otimes \partial \otimes 1}- a_i \otimes [a_j {}_{\mu} b_i]^{2}\otimes b_j\mid_{\mu = 1\otimes 1 \otimes \partial} - a_i \otimes a_j \otimes [b_j {}_{\mu} b_i]^{2}\mid_{\mu = 1\otimes \partial \otimes 1}\Big).\label{YB2}
\end{align}

\begin{lemm}\label{lem:mixed-coboundary-expansion}
		Let $(\A,[\cdot_\la \cdot]^{1},[\cdot_\la \cdot]^{2})$ be a compatible Lie conformal algebra and let $r = \sum_i a_i \otimes b_i\in \A \otimes \A$. Define $\de_1$ and $\de_2$  as in  \eqref{coboundary1} and \eqref{coboundary2}, respectively.   
		Then
		\begin{align}
		(\de_1 \otimes 1)(\de_2 (x)) &= \sum_{i,j}\Big([[x {}_{\lambda} a_i]^{2}{}_{\mu} a_j]^{1}\otimes b_j \otimes b_i + a_j \otimes [[x {}_{\lambda} a_i]^{2}{}_{\mu} b_j]^{1} \otimes b_i\nonumber\\
		&\quad+ [a_i {}_{\mu} a_j]^{1}\otimes b_j \otimes [x {}_{\lambda} b_i]^{2} + a_j \otimes [a_i {}_{\mu} b_j]^{1} \otimes [x {}_{\lambda} b_i]^{2}\Big)\Big|_{\lambda=-\partial^{\otimes 3},\:\:\mu=-(\partial^{\otimes 2} \otimes 1)},\label{mixed-coboundary-12}\\
		(\de_2 \otimes 1)(\de_1 (x)) &= \sum_{i,j}\Big([[x {}_{\lambda} a_i]^{1}{}_{\mu} a_j]^{2}\otimes b_j \otimes b_i + a_j \otimes [[x {}_{\lambda} a_i]^{1}{}_{\mu} b_j]^{2} \otimes b_i\nonumber\\
		&\quad+ [a_i {}_{\mu} a_j]^{2}\otimes b_j \otimes [x {}_{\lambda} b_i]^{1} + a_j \otimes [a_i {}_{\mu} b_j]^{2} \otimes [x {}_{\lambda} b_i]^{1}\Big)\Big|_{\lambda=-\partial^{\otimes 3},\:\:\mu=-(\partial^{\otimes 2} \otimes 1)}.\label{mixed-coboundary-21}
	\end{align}
\end{lemm}
\begin{proof}
		We have 
	\begin{align*}
		(\de_1 \otimes 1)(\de_2(x)) &= (\de_1 \otimes 1)\Big(\sum_i ([x_\la a_i]^{2}\otimes b_i + a_i \otimes [x_\la b_i]^{2})\big|_{\la=-\p^{\otimes 2}}\Big)\\
	                          &= (\de_1 \otimes 1)\Big(\sum_i \sum_{k\ge 0}\frac{(-1)^k}{k!}(\p \otimes 1 + 1 \otimes \p)^k((x_{(k)} a_i)^{2} \otimes b_i + a_i \otimes (x_{(k)} b_i)^{2}\Big).
	\end{align*}
	Using that $(\de_1 \otimes 1)(\p^{\otimes 2}(a\otimes b))$ = $\p^{\otimes 3}((\de_1 \otimes 1)(a\otimes b))$, we obtain 
	\begin{align*}
	(\de_1 \otimes 1)(\de_2(x)) &= \sum_i \sum_{k\ge 0}\frac{(-1)^k}{k!}(\partial^{\otimes 3})^k\left(\de_1((x_{(k)} a_i)^{2}) \otimes b_i + \de_1(a_i) \otimes (x_{(k)} b_i)^{2}\right)\\
		&= \sum_{i,j} \sum_{k\ge 0}\frac{(-1)^k}{k!}(\partial^{\otimes 3})^k\Big([(x_{(k)}a_i)^{2}{}_{\mu} a_j]^{1} \otimes b_j \otimes b_i + a_j \otimes [(x_{(k)}a_i)^{2}{}_{\mu} b_j]^{1} \otimes b_i\\
	&\quad+ [a_i {}_{\mu} a_j]^{1}\otimes b_j \otimes (x_{(k)}b_i)^{2} + a_j \otimes [a_i {}_{\mu} b_j]^{1} \otimes (x_{(k)}b_i)^{2}\Big)\Big|_{\mu = -(\partial^{\otimes 2}\otimes 1)}\\
		&= \sum_{i,j}\Big([[x {}_{\lambda} a_i]^{2}{}_{\mu} a_j]^{1}\otimes b_j \otimes b_i + a_j \otimes [[x {}_{\lambda} a_i]^{2}{}_{\mu} b_j]^{1} \otimes b_i\\
		&\quad+ [a_i {}_{\mu} a_j]^{1}\otimes b_j \otimes [x {}_{\lambda} b_i]^{2} + a_j \otimes [a_i {}_{\mu} b_j]^{1} \otimes [x {}_{\lambda} b_i]^{2}\Big)\Big|_{\lambda=-\partial^{\otimes 3},\:\:\mu=-(\partial^{\otimes 2} \otimes 1)}.
	\end{align*}
		This proves \eqref{mixed-coboundary-12}. Similarly, \eqref{mixed-coboundary-21} also holds.
\end{proof}

\begin{prop}\label{prop:mixed-YB-identity}
		Let $(\A,[\cdot_\la \cdot]^{1},[\cdot_\la \cdot]^{2})$ be a compatible Lie conformal algebra and let $r = \sum_i a_i \otimes b_i\in \A \otimes \A$. Define $\de_1$ and $\de_2$ as in \eqref{coboundary1} and \eqref{coboundary2}, respectively. Suppose that the symmetric part of $r$ is $\A$-invariant with respect to both brackets, that is, for every $u\in\A$,
	\begin{align}
		\Big(u_\la (r+r^{21})\Big)^{1}\mid_{\la = -\p^{\otimes 2}} = 0,\quad \Big(u_\la (r+r^{21})\Big)^{2}\mid_{\la = -\p^{\otimes 2}} = 0,\label{sym invariance}
	\end{align}
	where $r^{21} = \sum_i b_i \otimes a_i$. Then for every $x\in\A$,
	\begin{align}
		\sum_{c.p.}\Big((\de_1 \otimes 1)(\de_2(x)) + (\de_2 \otimes 1)(\de_1(x))\Big)
		+ (x_\la[\![r,r]\!]^{1})^{2}\mid_{\la = -\p^{\otimes 3}}
		+ (x_\la[\![r,r]\!]^{2})^{1}\mid_{\la = -\p^{\otimes 3}} = 0,\label{mixed YB identity}
	\end{align}
	where $\sum_{c.p.}$ applies to the whole expression inside the parentheses and means that we add the two nontrivial cyclic permutations of the tensor factors in $\A \otimes \A \otimes \A$.
\end{prop}
\begin{proof}
	Using Lemma \ref{lem:mixed-coboundary-expansion}, we have 
		\begin{align}
			\sum_{\text{c.p.}}&(\de_1 \otimes 1)(\de_2(x))\nonumber\\
			&= \sum_{i,j}\Big([[x {}_{\lambda} a_i]^{2}{}_{\mu} a_j]^{1}\otimes b_j \otimes b_i + a_j \otimes [[x {}_{\lambda} a_i]^{2}{}_{\mu} b_j]^{1} \otimes b_i\nonumber\\
			&\quad+ [a_{i}{}_{\mu} a_{j}]^{1} \otimes b_{j} \otimes [x{}_{\lambda} b_{i}]^{2} + a_{j} \otimes [a_{i} {}_{\mu} b_{j}]^{1} \otimes [x{}_{\lambda} b_{i}]^{2}\Big)\Big|_{\lambda = -\partial^{\otimes 3},\:\:\mu = -(\partial^{\otimes 2}\otimes 1)}\nonumber\\
			&\quad + \sum_{i,j}\Big(b_j \otimes b_i \otimes [[x {}_{\lambda} a_i]^{2}{}_{\mu} a_j]^{1} + [[x {}_{\lambda} a_i]^{2}{}_{\mu} b_j]^{1} \otimes b_i \otimes a_j\nonumber\\
			&\quad+ b_j \otimes [x{}_{\lambda} b_i]^{2} \otimes [a_i {}_{\mu} a_j]^{1} +[a_i {}_{\mu} b_j]^{1} \otimes [x{}_{\lambda} b_i]^{2} \otimes a_j\Big)\Big|_{\lambda = -\partial^{\otimes 3},\:\:\mu = -(1 \otimes 1 \otimes \partial + \partial \otimes 1 \otimes 1)}\nonumber\\
			&\quad + \sum_{i,j}\Big(b_i \otimes [[x {}_{\lambda} a_i]^{2}{}_{\mu} a_j]^{1} \otimes b_j + b_i \otimes a_j \otimes [[x {}_{\lambda} a_i]^{2}{}_{\mu} b_j]^{1}\nonumber\\
			&\quad+ [x{}_{\lambda} b_i]^{2} \otimes [a_i {}_{\mu} a_j]^{1} \otimes b_j + [x{}_{\lambda} b_i]^{2} \otimes a_j \otimes [a_i {}_{\mu} b_j]^{1}\Big)\Big|_{\lambda = -\partial^{\otimes 3},\:\:\mu = -(1\otimes \partial^{\otimes 2})},\label{compatible coadjoint 1}
		\end{align}
		\begin{align}
			\sum_{\text{c.p.}}&(\de_2 \otimes 1)(\de_1(x))\nonumber\\
			&= \sum_{i,j}\Big([[x {}_{\lambda} a_i]^{1}{}_{\mu} a_j]^{2} \otimes b_j \otimes b_i + a_j \otimes [[x {}_{\lambda} a_i]^{1}{}_{\mu} b_j]^{2} \otimes b_i \nonumber\\
			&\quad  + [a_i{}_{\mu} a_j]^{2} \otimes b_j \otimes [x{}_{\lambda} b_i]^{1} + a_j \otimes [a_i {}_{\mu} b_j]^{2} \otimes [x{}_{\lambda} b_i]^{1}\Big) \Big|_{\lambda = -\partial^{\otimes 3},\:\: \mu = -(\partial^{\otimes 2} \otimes 1)} \nonumber\\
			&\quad + \sum_{i,j}\Big(b_j \otimes b_i \otimes [[x {}_{\lambda} a_i]^{1}{}_{\mu} a_j]^{2} + [[x {}_{\lambda} a_i]^{1}{}_{\mu} b_j]^{2} \otimes b_i \otimes a_j \nonumber\\
			&\quad + b_j \otimes [x{}_{\lambda} b_i]^{1} \otimes [a_i {}_{\mu} a_j]^{2} + [a_i {}_{\mu} b_j]^{2} \otimes [x{}_{\lambda} b_i]^{1} \otimes a_j \Big) \Big|_{\lambda = -\partial^{\otimes 3}, \:\:\mu = -(1 \otimes 1 \otimes \partial + \partial \otimes 1 \otimes 1)} \nonumber\\
			&\quad + \sum_{i,j}\Big(b_i \otimes [[x {}_{\lambda} a_i]^{1}{}_{\mu} a_j]^{2} \otimes b_j + b_i \otimes a_j \otimes [[x {}_{\lambda} a_i]^{1}{}_{\mu} b_j]^{2} \nonumber\\
			&\quad + [x{}_{\lambda} b_i]^{1} \otimes [a_i {}_{\mu} a_j]^{2} \otimes b_j + [x{}_{\lambda} b_i]^{1} \otimes a_j \otimes [a_i {}_{\mu} b_j]^{2}\Big) \Big|_{\lambda = -\partial^{\otimes 3}, \:\:\mu = -(1 \otimes \partial^{\otimes 2})}\label{compatible coadjoint 2}
		\end{align}				
	We label the twelve summands in \eqref{compatible coadjoint 1} by
	$(1),\ldots,(12)$, reading from left to right and from the first cyclic block to the third cyclic block. Thus, for example, $(1)$ is
	$[[x_\la a_i]^{2}_\mu a_j]^{1}\otimes b_j \otimes b_i$,
	$(5)$ is $b_j \otimes b_i \otimes [[x {}_{\lambda} a_i]^{2}{}_{\mu} a_j]^{1}$, and
	$(9)$ is $b_i \otimes [[x {}_{\lambda} a_i]^{2}{}_{\mu} a_j]^{1} \otimes b_j$,
	each with the evaluation displayed in its cyclic block.
	The corresponding twelve summands in \eqref{compatible coadjoint 2} are denoted by
	$\tilde{(1)},\ldots,\tilde{(12)}$.

	On the other hand, we have 
		\begin{align}
			(x {}_{\lambda} [\![r,r]\!]^{2})^{1} \mid_{\lambda = -\partial^{\otimes 3}}
			&= \sum_{i,j}[ x {}_{\lambda}([a_i{}_{\mu} a_j]^{2}\otimes b_i \otimes b_j\mid_{\mu = 1 \otimes \partial \otimes 1})\nonumber\\
			& \quad - x {}_{\lambda} (a_i \otimes [a_j {}_{\mu} b_i]^{2} \otimes b_j\mid_{\mu = 1 \otimes 1\otimes \partial}) - x {}_{\lambda}(a_i \otimes a_j \otimes [b_j {}_{\mu} b_i]^{2}\mid_{\mu = 1\otimes \partial \otimes 1})]^{1}\mid_{\lambda = -\partial^{\otimes 3}}\nonumber\\
			&= \sum_{i,j}\Big([x {}_{\lambda} [a_i{}_{\mu} a_j]^{2}]^{1} \otimes b_i \otimes b_j\mid_{\mu = 1 \otimes \partial \otimes 1} + [a_i {}_{\mu} a_j]^{2} \otimes [x {}_{\lambda} b_i]^{1}\otimes b_j\mid_{\mu = -(\partial \otimes 1 \otimes 1 + 1 \otimes 1 \otimes \partial)}\nonumber\\
			& \quad +[a_i {}_{\mu} a_j]^{2} \otimes b_i \otimes [x {}_{\lambda} b_j]^{1}\mid_{\mu = 1 \otimes \partial \otimes 1} + [x {}_{\lambda} a_i]^{1} \otimes [b_i{}_{\mu} a_j]^{2} \otimes b_j\mid_{\mu = -(1 \otimes \partial \otimes 1 + 1 \otimes 1 \otimes \partial)}\nonumber\\
			& \quad + a_i \otimes [x {}_{\lambda} [b_i {}_{\mu} a_j]^{2}]^{1} \otimes b_j\mid_{\mu = -(1 \otimes \partial \otimes 1 + 1 \otimes 1 \otimes \partial)} + a_i \otimes [b_i {}_{\mu} a_j]^{2} \otimes [x {}_{\lambda} b_j]^{1}\mid_{\mu = \partial \otimes 1 \otimes 1}\nonumber\\
			& \quad + [x {}_{\lambda} a_i]^{1} \otimes a_j \otimes [b_i {}_{\mu} b_j]^{2}\mid_{\mu = -(1 \otimes \partial \otimes 1 + 1 \otimes 1 \otimes \partial)} + a_i \otimes [x {}_{\lambda} a_j]^{1} \otimes [b_i {}_{\mu} b_j]^{2}\mid_{\mu = \partial \otimes 1 \otimes 1}\nonumber\\
			& \quad + a_i \otimes a_j \otimes [x {}_{\lambda} [b_i {}_{\mu} b_j]^{2}]^{1}\mid_{\mu = -(1 \otimes \partial \otimes 1 + 1 \otimes 1 \otimes \partial)}\Big)\Big|_{\lambda = -\partial^{\otimes 3}}\label{compatible cybe 1},\\
			(x _{\lambda} [\![r,r]\!]^{1})^{2} \mid_{\lambda = -\partial^{\otimes 3}}
            &= \sum_{i,j}[x {}_{\lambda}([a_i{}_{\mu} a_j]^{1}\otimes b_i \otimes b_j\mid_{\mu = 1 \otimes \partial \otimes 1})\nonumber\\
            & \quad - x {}_{\lambda} (a_i \otimes [a_j {}_{\mu} b_i]^{1} \otimes b_j\mid_{\mu = 1 \otimes 1\otimes \partial}) - x {}_{\lambda}(a_i \otimes a_j \otimes [b_j {}_{\mu} b_i]^{1}\mid_{\mu = 1\otimes \partial \otimes 1})]^{2}\mid_{\lambda = -\partial^{\otimes 3}}\nonumber\\
            &= \sum_{i,j}\Big([x {}_{\lambda} [a_i{}_{\mu} a_j]^{1}]^{2} \otimes b_i \otimes b_j\mid_{\mu = 1 \otimes \partial \otimes 1} + [a_i {}_{\mu} a_j]^{1} \otimes [x {}_{\lambda} b_i]^{2}\otimes b_j\mid_{\mu = -(\partial \otimes 1 \otimes 1 + 1 \otimes 1 \otimes \partial)}\nonumber\\
            & \quad +[a_i {}_{\mu} a_j]^{1} \otimes b_i \otimes [x {}_{\lambda} b_j]^{2}\mid_{\mu = 1 \otimes \partial \otimes 1} + [x {}_{\lambda} a_i]^{2} \otimes [b_i{}_{\mu} a_j]^{1} \otimes b_j\mid_{\mu = -(1 \otimes \partial \otimes 1 + 1 \otimes 1 \otimes \partial)}\nonumber\\
            & \quad + a_i \otimes [x {}_{\lambda} [b_i {}_{\mu} a_j]^{1}]^{2} \otimes b_j\mid_{\mu = -(1 \otimes \partial \otimes 1 + 1 \otimes 1 \otimes \partial)} + a_i \otimes [b_i {}_{\mu} a_j]^{1} \otimes [x {}_{\lambda} b_j]^{2}\mid_{\mu = \partial \otimes 1 \otimes 1}\nonumber\\
            & \quad + [x {}_{\lambda} a_i]^{2} \otimes a_j \otimes [b_i {}_{\mu} b_j]^{1}\mid_{\mu = -(1 \otimes \partial \otimes 1 + 1 \otimes 1 \otimes \partial)} + a_i \otimes [x {}_{\lambda} a_j]^{2} \otimes [b_i {}_{\mu} b_j]^{1}\mid_{\mu = \partial \otimes 1 \otimes 1}\nonumber\\
            & \quad + a_i \otimes a_j \otimes [x {}_{\lambda} [b_i {}_{\mu} b_j]^{1}]^{2}\mid_{\mu = -(1 \otimes \partial \otimes 1 + 1 \otimes 1 \otimes \partial)}\Big)\Big|_{\lambda = -\partial^{\otimes 3}}\label{compatible cybe 2}.
		\end{align}
	We label the nine summands in \eqref{compatible cybe 1} by
	\circlednum{1},\ldots,\circlednum{9}, reading from left to right. Thus, for example,
	\circlednum{1} is $[x{}_\la[{a_i}_\mu a_j]^{2}]^{1}\otimes b_i\otimes b_j$,
	with the corresponding evaluation in $\la$ and $\mu$. The corresponding nine summands in
	\eqref{compatible cybe 2} are denoted by
	$\tilde{\circlednum{1}},\ldots,\tilde{\circlednum{9}}$.
	First, we have $(\la = -\p^{\otimes 3})$
	\begin{align*}
		(3) &= [a_i{}_{\mu} a_j]^{1} \otimes b_j \otimes [x{}_{\lambda} b_i]^{2}\mid_{\mu = -\partial \otimes 1 \otimes 1 - 1 \otimes \partial \otimes 1}\\
		    &= -[a_j{}_{-\mu-\partial \otimes 1 \otimes 1} a_i]^{1} \otimes b_j \otimes [x{}_{\lambda} b_i]^{2}\mid_{\mu = -\partial \otimes 1 \otimes 1 - 1 \otimes \partial \otimes 1}\\
			&= -[a_j{}_{\mu} a_i]^{1} \otimes b_j \otimes [x{}_{\lambda} b_i]^{2}\mid_{\mu = 1 \otimes \partial \otimes 1} = -\tilde{\circlednum{3}}.
	\end{align*}
		The same skew-symmetry calculation, with the two brackets interchanged and with the tensor
		positions cyclically permuted, gives
		$\tilde{(3)} + {\circlednum{3}} = 0$, $(4)+\tilde{\circlednum{6}} = 0$, and
		$\tilde{(4)} + \circlednum{6} = 0$.

		Interchanging the indices $i$ and $j$, using skew-symmetry, and then applying the compatible Jacobi identity to the triple $(x,a_i,a_j)$, we get
		\begin{align}
			&(1) + \circlednum{1} + \tilde{(1)} + \tilde{\circlednum{1}}\nonumber\\
			&= [[x {}_{\lambda} a_j]^{2}{}_{\mu} a_i]^{1}\otimes b_i \otimes b_j\mid_{\mu = -\partial^{\otimes 2}\otimes 1} + [x {}_{\lambda} [a_i{}_{\mu} a_j]^{2}]^{1} \otimes b_i \otimes b_j\mid_{\mu = 1\otimes \partial \otimes 1}\nonumber\\
			& \quad + [[x {}_{\lambda} a_j]^{1}{}_{\mu} a_i]^{2} \otimes b_i \otimes b_j\mid_{\mu = -\partial^{\otimes 2}\otimes 1} + [x {}_{\lambda} [a_i{}_{\mu} a_j]^{1}]^{2} \otimes b_i \otimes b_j\mid_{\mu = 1 \otimes \partial \otimes 1}\nonumber\\
			&= [[x {}_{\lambda} a_i]^{2}{}_{\mu} a_j]^{1} \otimes b_i \otimes b_j + [[x {}_{\lambda} a_i]^{1}{}_{\mu} a_j]^{2} \otimes b_i \otimes b_j\mid_{\mu = -(\partial \otimes 1 \otimes 1 + 1\otimes 1 \otimes \partial)}.\label{compatible cybe pf eq1}
		\end{align}
	Now, the invariance assumption \eqref{sym invariance} says explicitly that, for every $u\in\A$,
	\begin{align}
		\sum_i\Big(u_\la (a_i\otimes b_i + b_i \otimes a_i)\Big)^1\Big|_{\la = -\p^{\otimes 2}} = 0,\\
		\sum_i\Big(u_\la (a_i\otimes b_i + b_i \otimes a_i)\Big)^2\Big|_{\la = -\p^{\otimes 2}} = 0,
	\end{align}
	Together with \eqref{compatible cybe pf eq1}, this gives $(\la = -\p^{\otimes 3})$
	\begin{align*}
		&(6) + \tilde{(6)} + (1) + \circlednum{1} + \tilde{(1)} + \tilde{\circlednum{1}}\\
		&= ([[x_\la a_i]^{2}{}_\mu b_j]^{1} \otimes b_i \otimes a_j + [[x_\la a_i]^{1}{}_\mu b_j]^{2} \otimes b_i \otimes a_j + [[x_\la a_i]^{2}{}_\mu a_j]^{1} \otimes b_i \otimes b_j \\
		& \quad + [[x_\la a_i]^{1}{}_\mu a_j]^{2} \otimes b_i \otimes b_j)\mid_{\mu = -(\p \otimes 1 \otimes 1 + 1\otimes 1 \otimes \p)}\\
		&= (-a_j \otimes b_i \otimes [[x_\la a_i]^{2}{}_\mu b_j]^{1} - a_j \otimes b_i \otimes [[x_\la a_i]^{1}{}_\mu b_j]^{2} - b_j \otimes b_i \otimes [[x_\la a_i]^{2}{}_\mu a_j]^{1}\\
		& \quad - b_j \otimes b_i \otimes [[x_\la a_i]^{1}{}_\mu a_j]^{2} )\mid_{\mu = -(\p \otimes 1 \otimes 1 + 1\otimes 1 \otimes \p)}\\
		&:= (A) + \tilde{(A)} + (B) +\tilde{(B)}.
	\end{align*}
		With the same evaluation of $\lambda$ and $\mu$, we have $(B)+(5)=0$ and $\tilde{(B)}+\tilde{(5)}=0$. Thus it remains to cancel
		$(A)+\tilde{(A)}$.

	Now, using skew-symmetry and the invariance assumption \eqref{sym invariance}, we have 
	\begin{align*}
		&(11) + \tilde{\circlednum{4}}\\
		& = [x {}_{\lambda} b_i]^{2} \otimes [a_i {}_{\mu} a_j]^{1} \otimes b_j + 	[x {}_{\lambda} a_i]^{2} \otimes [b_i{}_{\mu} a_j]^{1} \otimes b_j\mid_{\mu = -(1 \otimes \partial \otimes 1 + 1 \otimes 1 \otimes \partial)}\\
		& = - [x {}_{\lambda} b_i]^{2} \otimes [a_j {}_{-\mu-1\otimes \partial \otimes 1}a_i]^{1} \otimes b_j - [x {}_{\lambda} a_i]^{2} \otimes [a_j {}_{-\mu-1\otimes \partial \otimes 1}b_i]^{1} \otimes b_j\mid_{\mu = -(1 \otimes \partial \otimes 1 + 1 \otimes 1 \otimes \partial)}\\
		& = -[(1 \otimes \ad_{1}(a_j)_{-\mu- 1\otimes \partial})([x {}_{\lambda} b_i]^{2}\otimes a_i + [x {}_{\lambda} a_i]^{2}\otimes b_i)\mid_{\lambda = -\partial^{\otimes 2}}]\otimes b_j\mid_{\mu = -1 \otimes \partial^{\otimes 2}}\\
		& = a_i \otimes [a_j {}_{-\mu-1 \otimes \partial \otimes 1}[x {}_{\lambda} b_i]^{2}]^{1} \otimes b_j + b_i \otimes [a_j {}_{-\mu-1\otimes \partial \otimes 1}[x {}_{\lambda} a_i]^{2}]^{1}\mid_{\lambda = -\partial^{\otimes 3}, \mu = -1\otimes \partial^{\otimes 2}}\\
		& = a_i \otimes [a_j {}_{\mu} [x {}_{\lambda} b_i]^{2}]^{1} \otimes b_j + b_i \otimes [a_j {}_{\mu} [x {}_{\lambda} a_i]^{2}]^{1} \otimes b_j\mid_{\lambda = -\partial^{\otimes 3},\; \mu = 1 \otimes 1 \otimes \partial}:= (C)+(D).
	\end{align*}
	Similarly, we have
	\begin{align*}
		\tilde{(11)} + \circlednum{4} &= a_i \otimes [a_j {}_{\mu} [x {}_{\lambda} b_i]^{1}]^{2} \otimes b_j + b_i \otimes [a_j {}_{\mu} [x {}_{\lambda} a_i]^{1}]^{2} \otimes b_j\mid_{\lambda = -\partial^{\otimes 3},\;\mu = 1\otimes 1\otimes \partial}\\
									  &:= \tilde{(C)} + \tilde{(D)},\\
		(12) + \tilde{\circlednum{7}} &= a_i \otimes a_j \otimes [b_j {}_{\mu} [x {}_{\lambda} b_i]^{2}]^{1} + b_i \otimes a_j \otimes [b_j {}_{\mu} [x {}_{\lambda} a_i]^{2}]^{1}\mid_{\lambda = -\partial^{\otimes 3},\;\mu = 1\otimes \partial \otimes 1}\\
									  &:= (E) + (F),\\
		\tilde{(12)} + \circlednum{7} &= a_i \otimes a_j \otimes [b_j {}_{\mu} [x {}_{\lambda} b_i]^{1}]^{2} + b_i \otimes a_j \otimes [b_j {}_{\mu} [x {}_{\lambda} a_i]^{1}]^{2}\mid_{\lambda = -\partial^{\otimes 3},\;\mu = 1\otimes \partial \otimes 1}\\
									  &:= \tilde{(E)} + \tilde{(F)},\\
		(8) + \tilde{\circlednum{2}} &= -b_j \otimes [x {}_{\lambda} b_i]^{2} \otimes [a_i {}_{\mu} a_j]^{1} - a_j \otimes [x {}_{\lambda} b_i]^{2} \otimes [a_i {}_{\mu} b_j]^{1}\mid_{\lambda = -\partial^{\otimes 3},\;\mu = -(\partial \otimes 1 \otimes 1 + 1 \otimes 1 \otimes \partial)}\\
									 &:= (G) + (H),\\
		\tilde{(8)} + \circlednum{2} &= -b_j \otimes [x {}_{\lambda} b_i]^{1} \otimes [a_i {}_{\mu} a_j]^{2} - a_j \otimes [x {}_{\lambda} b_i]^{1} \otimes [a_i {}_{\mu} b_j]^{2}\mid_{\lambda = -\partial^{\otimes 3},\;\mu = -(\partial \otimes 1 \otimes 1 + 1 \otimes 1 \otimes \partial)}\\
									 &:= \tilde{(G)} + \tilde{(H)},							 						  
	\end{align*}
		By termwise cancellation under the same evaluation of $\lambda$ and $\mu$, we have
		$(D) + (9) = 0$, $\tilde{(D)} + \tilde{(9)} = 0$, $(F) + (10) = 0$,
		$\tilde{(F)} + \tilde{(10)} = 0$, $(G) + (7) = 0$, and
		$\tilde{(G)} + \tilde{(7)} = 0$. Hence it remains only
		$(C)$, $\tilde{(C)}$, $(E)$, $\tilde{(E)}$, $(H)$ and $\tilde{(H)}$.

		After rewriting all brackets with the common evaluation
		$\lambda=-\partial^{\otimes 3}$ and
		$\mu=1\otimes1\otimes\partial$, and after using skew-symmetry to put the entries in the order
			$(a_j,x,b_i)$, the sum
				$(2) + \tilde{(2)} + \circlednum{5} + \tilde{\circlednum{5}} + (C) + \tilde{(C)}$
				is the compatible Jacobi identity for the triple $(a_j,x,b_i)$, with the
				outside tensor factors $a_i$ and $b_j$. Hence this sum is zero. Now, using skew-symmetry and
			the invariance assumption \eqref{sym invariance}, we get 
		\begin{align}\label{compatible cybe pf eq2}
			\circlednum{8} &+ \tilde{\circlednum{8}} + (A) + \tilde{(A)} + (H) + \tilde{(H)}\nonumber\\
			&= \Big(a_i \otimes [x {}_{\lambda} a_j]^{1} \otimes [b_i {}_{\mu} b_j]^{2} + a_i \otimes [x {}_{\lambda} a_j]^{2} \otimes [b_i {}_{\mu} b_j]^{1}\Big)\Big|_{\lambda = -\partial^{\otimes 3},\;\mu = \partial \otimes 1 \otimes 1}\nonumber\\
			& \quad +\Big(-a_j \otimes b_i \otimes [[x {}_{\lambda} a_i]^{2}{}_{\mu} b_j]^{1} - a_j \otimes b_i \otimes [[x {}_{\lambda} a_i]^{1}{}_{\mu} b_j]^{2}\Big)\Big|_{\lambda = -\partial^{\otimes 3},\;\mu = -(\partial \otimes 1 \otimes 1 + 1 \otimes 1 \otimes \partial)}\nonumber\\
			& \quad +\Big(-a_j \otimes [x {}_{\lambda} b_i]^{2} \otimes [a_i {}_{\mu} b_j]^{1} - a_j \otimes [x {}_{\lambda} b_i]^{1} \otimes [a_i {}_{\mu} b_j]^{2}\Big)\Big|_{\lambda = -\partial^{\otimes 3},\;\mu = -(\partial \otimes 1 \otimes 1 + 1 \otimes 1 \otimes \partial)}\nonumber\\
			&= a_i \otimes\Big([x {}_{\lambda} a_j]^{1} \otimes [b_i {}_{\mu-1\otimes \partial \otimes 1} b_j]^{2} + [x {}_{\lambda} a_j]^{2} \otimes [b_i {}_{\mu-1\otimes \partial \otimes 1} b_j]^{1} - b_j \otimes [[x {}_{\lambda} a_j]^{2} {}_{\mu} b_i]^{1}\nonumber\\
			& \quad - b_j \otimes [[x {}_{\lambda} a_j]^{1} {}_{\mu} b_i]^{2} - [x {}_{\lambda} b_j]^{2} \otimes [a_j {}_{\mu} b_i]^{1} - [x {}_{\lambda} b_j]^{1} \otimes [a_j {}_{\mu} b_i]^{2}\Big)\Big|_{\substack{\lambda = -\partial^{\otimes 3},\\ \mu = -(\partial \otimes 1 \otimes 1 + 1 \otimes 1 \otimes \partial)}}\nonumber\\
			&= a_i \otimes (1 \otimes \ad_{1}(b_i) {}_{\mu})\Big([x {}_{\lambda} a_j]^{2}\otimes b_j + b_j \otimes [x {}_{\lambda} a_j]^{2} + [x {}_{\lambda} b_j]^{2} \otimes a_j\Big)\Big|_{\lambda = -\partial^{\otimes 2},\:\mu = \partial \otimes 1 \otimes 1}\nonumber\\
			& \quad + a_i \otimes (1 \otimes \ad_{2}(b_i) {}_{\mu})\Big([x {}_{\lambda} a_j]^{1}\otimes b_j + b_j \otimes [x {}_{\lambda} a_j]^{1} + [x {}_{\lambda} b_j]^{1} \otimes a_j\Big)\Big|_{\lambda = -\partial^{\otimes 2},\:\mu = \partial \otimes 1 \otimes 1}\nonumber\\
			&= -a_i \otimes (1 \otimes \ad_{1}(b_i) {}_{\mu})(a_j \otimes [x {}_{\lambda} b_j]^{2}) - a_i \otimes (1 \otimes \ad_{2}(b_i)
            {}_{\mu})(a_j \otimes [x {}_{\lambda} b_j]^{1})\Big|_{\substack{\lambda = -\partial^{\otimes 3},\\ \mu = \partial \otimes 1 \otimes 1}}.
		\end{align}
		Finally, rewriting all brackets with
		the common evaluation $\lambda=-\partial^{\otimes 3}$ and
        $\mu=\partial\otimes1\otimes1$, gives the compatible Jacobi identity for the triple
		$(b_i,x,b_j)$. Hence $\circlednum{9}+\tilde{\circlednum{9}}+(E)+\tilde{(E)}+\eqref{compatible cybe pf eq2} = 0$.
		The proof is complete.
\end{proof}

	\begin{theo}\label{com YB}
		Let $(\A,[\cdot _\la \cdot]^{1},[\cdot _\la \cdot]^{2})$ be a compatible Lie conformal algebra and let $r\in \A \otimes \A$. Let 
		\begin{align*}
			\de_1(a) &= (a_\la r)^{1}\mid_{\la = -\p^{\otimes 2}}
			= \sum_i ([a_\la a_i]^{1}\otimes b_i + a_i \otimes [a_\la b_i]^{1})
			\Big|_{\la = -\p^{\otimes 2}},\\
			\de_2(a) &= (a_\la r)^{2}\mid_{\la = -\p^{\otimes 2}}
			= \sum_i ([a_\la a_i]^{2}\otimes b_i + a_i \otimes [a_\la b_i]^{2})
			\Big|_{\la = -\p^{\otimes 2}}.
		\end{align*}
	Then $(\A,[\cdot_\la\cdot]^1,[\cdot_\la\cdot]^2,\de_1,\de_2)$ is a compatible Lie conformal bialgebra if and only if the following conditions are satisfied:
    \begin{itemize}[(ii)] 
	\item[(i)] the symmetric part of $r$ is $\A$-invariant with respect to both brackets, that is, for every $a\in\A$,
	\begin{align*}
		\Big(a_\la (r+r^{21})\Big)^{1}\mid_{\la = -\p^{\otimes 2}} = 0,\quad \Big(a_\la (r+r^{21})\Big)^{2}\mid_{\la = -\p^{\otimes 2}} = 0,
	\end{align*}
	where $r^{21} = \sum_i b_i \otimes a_i$, if $r = \sum_i a_i \otimes b_i$.
	\item[(ii)] for every $a\in\A$,
	\begin{align}
		(a_\la[\![r,r]\!]^{1})^{1}\mid_{\la = -\p^{\otimes 3}} &= 0,\label{comYB first}\\
		(a_\la[\![r,r]\!]^{2})^{2}\mid_{\la = -\p^{\otimes 3}} &= 0,\label{comYB second}\\
		(a_\la[\![r,r]\!]^{1})^{2}\mid_{\la = -\p^{\otimes 3}} &+ (a_\la[\![r,r]\!]^{2})^{1}\mid_{\la = -\p^{\otimes 3}} = 0.\label{comYB third}
	\end{align} 
	where $\p^{\otimes 3} = \p \otimes 1 \otimes 1 + 1 \otimes \p \otimes 1 + 1 \otimes 1 \otimes \p$.
    \end{itemize}
\end{theo}
\begin{proof}
	Since $[\cdot_\la \cdot]^{1}$ and $[\cdot_\la \cdot]^{2}$ are compatible, the sum bracket is again a
	Lie conformal bracket. For this sum bracket, the coboundary map defined by $r$ is $\de_1+\de_2$,
	and the corresponding Yang--Baxter tensor is
	$[\![r,r]\!]^{1}+[\![r,r]\!]^{2}$.

	By Theorem \ref{thm:Li-coboundary}, applied separately to the two brackets, the two individual
	coboundary Lie conformal bialgebra conditions are equivalent precisely to the invariant
	symmetric-part condition for both brackets and to \eqref{comYB first}, \eqref{comYB second}.
	For the sum bracket, the invariant symmetric-part condition is obtained by adding the two
	invariant identities above; after \eqref{comYB first} and \eqref{comYB second}, the
		conformal Yang--Baxter condition for the sum bracket is therefore equivalent to \eqref{comYB third}.

	It remains to identify the compatible co-Jacobi condition. The individual bialgebra conditions
	imply that $\de_1$ and $\de_2$ are skew-symmetric. Hence \eqref{Compatible Co}
	is equivalent to
	\[
		\sum_{c.p.}\Big((\de_1 \otimes 1)(\de_2(x)) + (\de_2 \otimes 1)(\de_1(x))\Big)=0,\qquad x\in\A.
	\]
	By Proposition \ref{prop:mixed-YB-identity}, the compatible cobracket identity \eqref{Compatible Co} is equivalent to
	\eqref{comYB third}. Combining these equivalences with Definition \ref{com bia} proves the theorem.
\end{proof}
\begin{remark}
	We shall refer to \eqref{comYB first} and \eqref{comYB second} as the conformal Yang-Baxter condition for $(\A,[\cdot_\la\cdot]^1)$ and $(\A,[\cdot_\la\cdot]^2)$,respectively. And we call \eqref{comYB third} compatible conformal Yang-Baxter condition.
\end{remark}

\begin{ex}\label{ex-W-coboundary-compatible}
		Let $(\A,[\cdot_\lambda\cdot]^1,[\cdot_\lambda\cdot]^2)$ be the compatible Lie conformal algebra
		in Example \ref{ex-W-dual-coalgebra}.
	Consider
	\[
		r=H\otimes \p H-\p H\otimes H.
	\]
	Then $r+r^{21}=0$, so the invariant symmetric part condition in Theorem \ref{com YB} is automatically satisfied for both brackets.
	
	Moreover, every tensor factor of $r$ lies in the abelian conformal submodule $\C[\p]H$ for
	both $W(b_s)$ structures. Therefore each term in the conformal Yang--Baxter expressions contains
	a bracket between two elements of $\C[\p]H$, and hence
		$[\![r,r]\!]^1=[\![r,r]\!]^2=0$.
	Thus \eqref{comYB first}-\eqref{comYB third} hold, and $r$ defines a coboundary compatible Lie conformal bialgebra.
	
	Let $A_s(\la)=[L_\la H]^s=(\p+(1-b_s)\la)H$, $s=1,2$. Then $[L_\la \p H]^s=(\p+\la)A_s(\la)$.
	The corresponding cobrackets $\de_s(a)=(a_\la r)^s|_{\la=-\p^{\otimes 2}}$ are $\de_s(H)=0$
	and
	\begin{align}
		\de_s(L)
		&=\Big(A_s(\la)\otimes \p H+H\otimes (\p+\la)A_s(\la)
		-(\p+\la)A_s(\la)\otimes H-\p H\otimes A_s(\la)\Big)\Big|_{\la=-\p^{\otimes 2}}.\label{W-coboundary-delta}
	\end{align}
	In general $\de_s(L)$ is nonzero. Consequently
		$(\A,[\cdot_\la\cdot]^1,[\cdot_\la\cdot]^2,\de_1,\de_2)$
	is a nontrivial coboundary compatible Lie conformal bialgebra.
\end{ex}

\begin{defi}
	Let $(\A,[\cdot_\la\cdot]^1,[\cdot_\la\cdot]^2)$ be a compatible LCA and $r\in\A\otimes\A$. The compatible conformal CYBE is given by the following equations
	\[
	[\![r,r]\!]^{1}\equiv0\pmod{\partial^{\otimes3}},
	\qquad
	[\![r,r]\!]^{2}\equiv0\pmod{\partial^{\otimes3}}.
	\]
	We call $r$ a solution of the compatible conformal CYBE.
\end{defi}

	We now clarify the relationship between compatible conformal CYBE and compatible conformal Yang-Baxter condition.

\begin{lemm}\label{SCYBE0}
	Let $\A$ be a Lie conformal algebra and let $X\in \A\otimes\A\otimes\A$. If
	$X\equiv0\pmod{\partial^{\otimes3}}$,
	then, for every $a\in\A$,
	\[
		a_\lambda X\mid_{\lambda=-\partial^{\otimes3}}=0.
	\]
	Here $a_\lambda X$ is taken with respect to the induced action on $\A^{\otimes3}$.
\end{lemm}
\begin{proof}
	Let $X=\partial^{\otimes3}Y$, $Y\in\A^{\otimes3}$. By conformal sesquilinearity for the induced action on
	$\A^{\otimes3}$,
	\[
		a_\lambda X
		=a_\lambda(\partial^{\otimes3}Y)
		=(\lambda+\partial^{\otimes3})a_\lambda Y.
	\]
	After the evaluation $\lambda=-\partial^{\otimes3}$, the right-hand side is zero.
\end{proof}

		If $r$ is a solution of the compatible conformal CYBE, then
	$[\![r,r]\!]^1\equiv[\![r,r]\!]^2\equiv0\pmod{\partial^{\otimes3}}$, so by
	Lemma \ref{SCYBE0} the three conformal Yang--Baxter conditions
	\eqref{comYB first}--\eqref{comYB third} all hold. However, \eqref{comYB first}
	and \eqref{comYB second} do not imply \eqref{comYB third}, and
	\eqref{comYB first}--\eqref{comYB third} do not imply that $r$ is a solution of
	the compatible conformal CYBE. The following two examples illustrate these two facts,
	respectively.

\begin{ex}\label{ex-weak-YB-not-automatic}
	Let $\A=\C[\partial]L\oplus \C[\partial]H$.
	Define two $\lambda$-brackets by
	\[
		[L_\lambda L]^s=(\partial+2\lambda)L,\qquad [H_\lambda H]^s=0,\qquad s=1,2,
	\]
	and
	\[
		[L_\lambda H]^1=\partial H,\qquad [L_\lambda H]^2=(\partial-\lambda)H,
	\]
	with the remaining brackets determined by skew-symmetry and conformal sesquilinearity.
	These are the Lie conformal algebras $W(1)$ and $W(2)$, respectively; hence they are
	compatible by Example \ref{exW1}. Take
	\[
		r=L\otimes H-H\otimes L=L\wedge H.
	\]
	Let $\partial_i$ denote the action of $\partial$ on the
	$i$-th tensor factor. A direct calculation gives
	\[
		[\![r,r]\!]^1
		=\p^{\otimes3}\big(L\otimes H\otimes H-H\otimes H\otimes L+H\otimes L\otimes H\big),
	\]
	and
	\begin{align*}
		[\![r,r]\!]^2
		&=(\partial_1+2\partial_3)L\otimes H\otimes H
		+(\partial_2-\partial_1)H\otimes H\otimes L\\
		&\quad +(2\partial_1+\partial_2)H\otimes L\otimes H.
	\end{align*}
	One checks that
	\[
		(a_\lambda[\![r,r]\!]^1)^1\mid_{\lambda=-\p^{\otimes3}}=0,\qquad
		(a_\lambda[\![r,r]\!]^2)^2\mid_{\lambda=-\p^{\otimes3}}=0,\qquad a=L,H.
	\]
	Thus \eqref{comYB first} and \eqref{comYB second} hold. However, for $a=L$,
	\begin{align*}
		& (L_\lambda[\![r,r]\!]^1)^2\mid_{\lambda=-\p^{\otimes3}}
		+(L_\lambda[\![r,r]\!]^2)^1\mid_{\lambda=-\p^{\otimes3}}\\
		&\quad =2\p^{\otimes3}\Big((\partial_2-\partial_3)L\otimes H\otimes H
		+(\partial_1-\partial_2)H\otimes H\otimes L\\
		&\qquad\quad +(\partial_3-\partial_1)H\otimes L\otimes H\Big)\neq0
	\end{align*}
	in $\A^{\otimes3}$. Hence \eqref{comYB third} does not follow from
	\eqref{comYB first} and \eqref{comYB second}.
\end{ex}

\begin{ex}\label{ex-sv-nonzero-mixed-cybe}
		Let $(\A,[\cdot_\la\cdot]^1,[\cdot_\la\cdot]^2)$ be the compatible Lie conformal algebra
		in Example \ref{ex-sv-compatible-LCA}. Consider the skew-symmetric tensor
	\[
		r=P\otimes Q-Q\otimes P=P\wedge Q.
	\]
	Then $r+r^{21}=0$, so the invariant symmetric part condition in Theorem \ref{com YB} holds for both brackets.
	A direct calculation gives
	\[
		[\![r,r]\!]^1=P\wedge Q\wedge M_1\not\equiv0
		\pmod{\partial^{\otimes3}},\qquad
		[\![r,r]\!]^2=P\wedge Q\wedge M_2\not\equiv0
		\pmod{\partial^{\otimes3}},
	\]
	where $U\wedge V\wedge W$ denotes the alternating tensor in $\A^{\otimes3}$, namely
	\begin{align*}
		U\wedge V\wedge W
		&=U\otimes V\otimes W+V\otimes W\otimes U+W\otimes U\otimes V\\
		&\quad-U\otimes W\otimes V-W\otimes V\otimes U-V\otimes U\otimes W.
	\end{align*}
	
	Let $\Omega_i=P\wedge Q\wedge M_i$, $i=1,2$. Since $[\![r,r]\!]^s=\Omega_s$, to verify
	\eqref{comYB first} and \eqref{comYB second} it suffices to compute
	$(a_\lambda\Omega_s)^s$ for $s=1,2$ and for the generators $a$ of $\A$. For each $s$, the actions
	of $P$ and $Q$ produce alternating tensors with two equal factors $M_s$, while the actions of
	$M_1$ and $M_2$ are zero. For the action of $L$ on $\Omega_s$, the three tensor positions contribute
	\[
		\left(\p\otimes1\otimes1+\frac{1}{2}\la\right)
		+\left(1\otimes\p\otimes1+\frac{1}{2}\la\right)
		+1\otimes1\otimes\p
		=\p^{\otimes3}+\la,
	\]
	which vanishes after the evaluation $\la=-\p^{\otimes3}$. Hence \eqref{comYB first} and \eqref{comYB second} hold for all $a\in\A$.

	We now verify the mixed case, namely \eqref{comYB third}.
	For $a=L$, the two terms in \eqref{comYB third} are
	$(\partial^{\otimes3}+\lambda)\Omega_1$ and
	$(\partial^{\otimes3}+\lambda)\Omega_2$, respectively, and hence both vanish after the
	evaluation $\lambda=-\partial^{\otimes3}$. For $a=M_i$, $i=1,2$, the two terms in \eqref{comYB third} are zero.
	For $a=P$, one has
	\[
		(P_\lambda[\![r,r]\!]^1)^2=-P\wedge M_1\wedge M_2,\qquad
		(P_\lambda[\![r,r]\!]^2)^1=P\wedge M_1\wedge M_2.
	\]
	For $a=Q$, one has
	\[
		(Q_\lambda[\![r,r]\!]^1)^2=-Q\wedge M_1\wedge M_2,\qquad
		(Q_\lambda[\![r,r]\!]^2)^1=Q\wedge M_1\wedge M_2.
	\]
	Thus \eqref{comYB third} holds for every $a\in\A$.
	By Theorem \ref{com YB}, $r$ defines a coboundary compatible Lie conformal bialgebra.
	
	The corresponding cobrackets are nonzero. More precisely,
	\[
		\de_s(L)=0,\qquad \de_s(M_1)=\de_s(M_2)=0,\qquad s=1,2,
	\]
	and
	\[
		\de_s(P)=P\wedge M_s,\qquad \de_s(Q)=Q\wedge M_s,\qquad s=1,2.
	\]
		Thus this example shows that \eqref{comYB first}-\eqref{comYB third} may hold
		even when both $[\![r,r]\!]^1$ and $[\![r,r]\!]^2$ are nonzero; in this case
		\eqref{comYB third} is satisfied by cancellation of nonzero terms.
	\end{ex}


\section{Conclusion}

In this paper, we introduced compatible Lie conformal bialgebras and developed their basic 
structural theory. We first studied compatible Lie conformal algebras, their representations and 
matched pairs, and then established the corresponding duality between finite compatible Lie 
conformal algebras and compatible Lie conformal coalgebras. Based on these preparations, we proved 
that compatible Lie conformal bialgebras, standard compatible conformal Manin triples and matched pairs of 
compatible Lie conformal algebras are equivalent in the finite-rank case. For the coboundary case, we characterized the tensors $r$ that define compatible Lie conformal bialgebras by the invariance of $r+r^{21}$ with respect to both brackets and three conformal Yang--Baxter conditions. We introduced the compatible conformal CYBE for comparison and proved that each of its solutions satisfies these three conditions. The converse does not hold. The examples show that the first two conditions do not imply the third and that all three may hold even when $r$ is not a solution of the compatible conformal CYBE; in the latter case, the third condition holds by cancellation of nonzero terms.

These results extend the classical bialgebra--Manin triple--Yang--Baxter correspondence to the 
compatible conformal setting. They also provide an algebraic framework for constructing new Lie 
conformal bialgebra structures from compatible brackets, compatible cobrackets and $r$-matrices. 
Further directions include the classification of finite examples, the study of explicit solutions 
of the compatible conformal Yang--Baxter equation and the investigation of connections with 
deformation theory, vertex algebra theory and integrable systems.

\section*{Statements and Declarations}

\subsection*{Funding}
The authors received no specific funding for this work.

\subsection*{Competing interests}
The authors declare that they have no competing interests.

\subsection*{Data availability}
Data sharing is not applicable to this article as no datasets were generated or analysed during 
the current study.

\end{document}